%% file: preprint.tex
\documentclass[letterpaper]{article}

\usepackage[preprint,nonatbib]{neurips_2024}

\usepackage{amsmath,amssymb,amsthm}
\usepackage{hyperref}
\usepackage{theoremref}
\usepackage{mathtools}
\usepackage{algorithm}
\usepackage{algpseudocodex}
\usepackage{subcaption}
\usepackage{tikz}
\usepackage{multirow}
\usepackage{booktabs}
\usepackage{adjustbox}

\newtheorem{theorem}{Theorem}
\newtheorem{lemma}{Lemma}

\newtheorem{corollary}{Corollary}

\newtheorem{definition}{Definition}
\newtheorem{example}{Example}

\newcommand{\argmax}{\mathop{\rm argmax}\limits}

\usetikzlibrary{arrows}
\usetikzlibrary{shapes.geometric}
\usetikzlibrary{positioning}
\usetikzlibrary{intersections, calc, arrows}

\makeatletter
\newcommand{\figcaption}[1]{\def\@captype{figure}\caption{#1}}
\newcommand{\tblcaption}[1]{\def\@captype{table}\caption{#1}}
\makeatother

\title{Socially efficient mechanism on the minimum budget}

\author{
    \textbf{Hirota Kinoshita}\\
    The University of Tokyo\thanks{This research project was conducted at IBM Research.}\\
    \texttt{hirotak@g.ecc.u-tokyo.ac.jp}
    \and
    \textbf{Takayuki Osogami}\\
    IBM Research - Tokyo\\
    \texttt{osogami@jp.ibm.com}
    \and
    \textbf{Kohei Miyaguchi}\\
    IBM Research - Tokyo\\
    \texttt{miyaguchi@ibm.com}
}

\begin{document}

\maketitle

\input{abstract.tex}

\input{introduction.tex}

\input{model.tex}

\input{vcg.tex}

\input{proposed.tex}

\input{experiments.tex}

\input{conclusion.tex}

\clearpage
\bibliographystyle{plain}
\bibliography{reference}


\clearpage
\appendix

\input{proof.tex}

\input{examples.tex}

\input{duality.tex}

\input{ama.tex}

\input{more_exp.tex}

\input{redistribution.tex}
\input{impact.tex}

\end{document}

%% file: abstract.tex
\begin{abstract}
    In social decision-making among strategic agents,
    a universal focus lies on the balance between social and individual interests.
    Socially efficient mechanisms are thus desirably designed to not only maximize the social welfare but also incentivize the agents for their own profit.
    Under a generalized model that includes applications such as double auctions and trading networks,
    this study establishes a socially efficient (SE),
    dominant-strategy incentive compatible (DSIC),
    and individually rational (IR) mechanism with the minimum total budget expensed to the agents.
    The present method exploits discrete and known type domains to reduce a set of constraints into the shortest path problem in a weighted graph.
    In addition to theoretical derivation,
    we substantiate the optimality of the proposed mechanism through numerical experiments,
    where it certifies strictly lower budget than Vickery-Clarke-Groves (VCG) mechanisms for a wide class of instances.
\end{abstract}

%% file: introduction.tex
\section{Introduction} \label{sec:intro}

We address the ubiquitous problem of getting a group of self-interested strategic agents to act in the best interests of society.  A common approach in mechanism design is to introduce a broker, who selects the optimal option with respect to social welfare, maximizing the total benefits for the agents or minimizing their total costs, and incentivizes the agents to act for social welfare according to the selected option by paying them or receiving payments from them.  Such mechanisms that can be formally characterized by the properties of social efficiency (SE), dominant-strategy incentive compatibility (DSIC), and individually rationality (IR) are considered to be effective in various applications such as
trading networks \cite{myerson1983efficient,Osogami2023}
and
double-sided auction in cloud markets \cite{chichin2017cloud,jiang2018,kumar2018cloud}.
In particular,
a powerful generic solution is the Vickrey-Clarke-Groves (VCG) mechanism,
which has been analyzed and applied in numerous technical and practical contexts.

However, the prior work establishes impossibility theorems, suggesting that SE, DSIC, and IR may not be achieved without a budget contributed from the broker \cite{myerson1983efficient,Osogami2023}.  It would thus be reasonable for the broker to minimize the budget while achieving the three properties or to maximize revenue (i.e., negative budget) if those properties can be achieved without a (positive) budget, possibly to compensate for the budget needed for other occasions.
An example appears in workforce management,
where each labor receives compensation for their workload from the supervisor (broker) who seeks to achieve a goal as a result of the collective work;
then the supervisor should minimize the total compensation as long as the goal is achieved with minimal total workload (negative social welfare is minimized).
Related applications include federated learning \cite{zhan2022survey} and cloud-sourcing \cite{xia2017revenue}.

Our primary contribution is a novel algorithm that gives the budget-optimal solution among all mechanisms that satisfy the desired properties of SE, DISC, and IR.  A key idea in our approach is to exploit discrete and known type domains of agents to reduce a set of constraints, which are required by the above properties, into the shortest path problem in a weighted graph.  Unlike the majority of related works, the proposed mechanism lies beyond the celebrated class of VCG mechanisms.  The proposed mechanism outperforms the best possible VCG mechanism and gives strictly lower budgets for a majority of random instances in our numerical experiments.

The present work develops upon a general model of mechanism design formally defined in Section~\ref{sec:model}.
As a baseline, we revisit the VCG mechanism in Section~\ref{sec:vcg} from the perspective of minimizing the budget.
In Section~\ref{sec:proposed},
we develop a mechanism that achieves the minimum budget among all mechanisms having the desired properties,
followed by theoretical discussions about the optimality and the computational complexity.
Finally,
numerical experiments in Section~\ref{sec:exp} add empirical analysis of the optimal budget obtained by the proposed mechanism.

\section{Related work} \label{sec:related}

\paragraph{Maximizing revenue without social efficiency}

There have been a number of studies that seek to maximize the revenue just as we minimize the budget,
however, at the expense of SE (thus only guarantee DSIC and IR).
They mostly focus on one-sided auctions,
where the agents incur no costs and hence the broker charges them.
The noticeable majority of such approaches study the class of affine maximizer auction (AMA) \cite{Likhodedov2004,Likhodedov2005,Sandholm2015,Guo2017,Duan2023,Curry2024},
which can be seen as a technical variant of the VCG mechanism.
While those studies are clearly distinguished from ours as they no longer require SE and mostly restrict themselves to auctions and to the specific class of solutions,
we discuss in Appendix~\ref{sec:ama} how the proposed mechanism can also be technically extended along this setting.

\paragraph{Balancing budget}

A series of studies have sought to control the budget exactly to zero (strong budget balance; SBB),
or suppress it at most zero (weak budget balance; WBB) instead.
On one hand,
they have revealed negative facts such as the Myerson-Satterthwaite theorem~\cite{myerson1983efficient,Othman2009},
the Green-Laffont impossibility theorem~\cite{green1977impossibility,green1979impossibility},
and other related or extended results \cite{schweizer2006impossibility,nath2016impossibility,Osogami2023}.
On the other hand, when WBB is achievable,
instead of maximizing the revenue as we do,
the prior work has investigated ways to distribute the positive revenue back to agents,
which has been an interest in socially efficient mechanism design since the earliest literature
\cite{Tideman1976,Groves1977,Tideman1977,Rob1982,Collinge1983,bailey1997demand}.
In a variety of related studies,
they have derived analytical methods for some simple cases
\cite{cavallo2006optimal,guo2007worst,guo2009worst,gujar2011redistribution,guo2011vcg,guo2012worst,Victor2013}
and also provided approximately optimal heuristics for more general difficult settings
\cite{dufton2021optimizing,manisha2018learning,tacchetti2022learning}.
Unlike these prior studies, our approach minimizes the budget in a provably optimal manner when WBB is not achievable, and maximizes the revenue when it is achievable.
The positive revenue achieved by our approach could also be redistributed to the agents, as we discuss in Appendix~\ref{sec:redist}.

%% file: model.tex
\section{Model} \label{sec:model}
We start with a common context in mechanism design \cite{Nisan2007},
formally encapsulated in \thref{def:env} as an environment.
It involves a set of agents $\mathcal{N}$ and a set of available options $\mathcal{X}$;
each agent $i\in\mathcal{N}$ gains individual value $v_i(X)\in\mathbb{R}$ determined upon each option $X\in\mathcal{X}$ by a type $v_i\colon\mathcal{X}\to\mathbb{R}$ drawn from their type domain $\mathcal{V}_i$.
As supplementary notations,
given an environment $\mathcal{E} = (\mathcal{N}, \mathcal{X}, \mathcal{V})$,
we let $\mathcal{V}_{-i} \coloneqq \prod_{j\in\mathcal{N}\setminus\{i\}} \mathcal{V}_{j}$ for each agent $i\in\mathcal{N}$
denote the Cartesian product of all agents' type domains except $i$'s;
similarly,
for any $v\in\mathcal{V}$ and $i\in\mathcal{N}$,
let $v_{-i} \coloneqq (v_j)_{j\in\mathcal{N}\setminus\{i\}} \in \mathcal{V}_{-i}$
denote a partial list of $v$ that excludes the $i$'s type $v_i$ only.

\begin{definition}[Environment] \thlabel{def:env}
    An \textbf{environment} denotes a tuple $\mathcal{E} = (\mathcal{N}, \mathcal{X}, \mathcal{V})$,
    which consists of a non-empty finite set of \textbf{agents} $\mathcal{N}$,
    a non-empty (possibly infinite) set of \textbf{options} $\mathcal{X}$,
    and the Cartesian product $\mathcal{V} = \prod_{i\in\mathcal{N}}\mathcal{V}_i$ of non-empty finite sets of \textbf{types},
    or \textbf{type domains},
    $\mathcal{V}_i\subseteq\mathbb{R}^\mathcal{X}$ for each agent $i\in\mathcal{N}$.
\end{definition}

A mechanism in \thref{def:AMD} is exercised upon the environment by introducing an independent party,
who we call the broker.
It serves as a decision maker on behalf of agents,
while it also compensates or charges them depending on the types reported by them, which may differ from their true types.
The present model treats the mechanism as a static and public protocol known to all agents.

\begin{definition}[Mechanism] \thlabel{def:AMD}
    For an environment $\mathcal{E}=(\mathcal{N}, \mathcal{X}, \mathcal{V})$,
    a \textbf{mechanism} $\mathcal{M} = (\phi, \tau)$ consists of
    an \textbf{option rule} $\phi:\mathcal{V}\rightarrow\mathcal{X}$
    and
    a \textbf{payment rule} $\tau:\mathcal{V}\rightarrow\mathbb{R}^{\mathcal{N}}$,
    which works as follows:
    \begin{enumerate}
        \item Each agent $i\in\mathcal{N}$ reports a type $\hat v_i\in\mathcal{V}_i$ to the broker.
        \item Based on the list of reported types $\hat v = (\hat v_i)_{i\in\mathcal{N}}\in\mathcal{V}$,
              the broker selects an option $\phi(\hat v)\in\mathcal{X}$ and
              pays $\tau_i(\hat v)\in\mathbb{R}$ to each agent $i\in\mathcal{N}$.\footnote{A negative payment $\tau_i(\hat v)<0$ means the opposite transaction from the agent $i$ to the broker.}
        \item It demands the \textbf{budget} $B(\hat{v};\mathcal{M})\coloneqq\sum_{i\in\mathcal{N}}\tau_i(\hat{v})$ of the broker,
              and results in the \textbf{utilities} $u_i(\hat{v}; \mathcal{M})\coloneqq v_i(\phi(\hat{v})) + \tau_i(\hat{v})$ for each agent $i\in\mathcal{N}$, where $v_i$ denotes the true type of the agent $i$.
    \end{enumerate}
\end{definition}

On designing such a mechanism,
a prioritized interest lies in the social welfare,
which amounts to the net benefit that the mechanism produces on the environment as a whole.
A mechanism is thereby expected to always select the option that maximizes the social welfare on behalf of the environment,
when it is said to satisfy Social Efficiency, or SE~\eqref{eq:SE}, in the following \thref{def:SE}.
In addition,
because agents individually seek higher profit,
mechanisms are required to be justified in the sense of~\thref{def:justified}.
Namely, they should satisfy not only SE but also both of two important properties in the following \thref{def:DSIC_IR}.
Firstly,
Dominant-Strategy Incentive Compatibility,
or DSIC~\eqref{eq:DSIC},
incentivizes agents to report their own type truthfully.
In other words,
DSIC does not allow
any agent to gain by reporting a false type.
Secondly,
Individual Rationality,
or IR~\eqref{eq:IR},
clearly ensures every agent non-negative utility;
they do not want to experience negative utility due to participation.
Note that our formal definitions of SE \eqref{eq:SE} and IR \eqref{eq:IR} correctly match the literal concepts when they are accompanied by DSIC \eqref{eq:DSIC}.
We basically consider mechanisms that satisfy DSIC hereafter,
and thus do not explicitly distinguish the types reported by agents from their true ones.

\begin{definition}[SE] \thlabel{def:SE}
    For an environment $\mathcal{E}=(\mathcal{N}, \mathcal{X}, \mathcal{V})$,
    a mechanism $\mathcal{M} = (\phi, \tau)$,
    or just an option rule $\phi$,
    is said to satisfy
    \textbf{Social Efficiency (SE)} if and only if
    \begin{align}
         & \phi(v) \in \argmax_{X\in\mathcal{X}} S(X; v), & \forall v\in\mathcal{V}, \label{eq:SE}
    \end{align}
    where we define the \textbf{social welfare} $S:\mathcal{X}\times\mathcal{V}\to\mathbb{R}$ as
    \begin{align}
        S(X; v) \coloneqq  \sum_{i\in\mathcal{N}} v_i(X). \label{eq:welfare}
    \end{align}
    The social welfare due to the mechanism $\mathcal{M} = (\phi, \tau)$ is also expressed as follows:
    \begin{align}
        S(\phi(v); v) = \sum_{i\in\mathcal{N}} u_i(v;\mathcal{M}) - B(v;\mathcal{M}). \label{eq:welfare_sum}
    \end{align}
\end{definition}

\begin{definition}[DSIC, IR] \thlabel{def:DSIC_IR}
    For an environment $\mathcal{E}=(\mathcal{N}, \mathcal{X}, \mathcal{V})$,
    a mechanism $\mathcal{M} = (\phi, \tau)$ is said to satisfy
    \textbf{Dominant-Strategy Incentive Compatibility (DSIC)} if and only if
    \begin{align}
        u_i(v; \mathcal{M}) & \ge v_i(\phi(v_i', v_{-i})) + \tau_i(v_i', v_{-i}), & \forall v_i'\in\mathcal{V}_i, \forall i\in\mathcal{N},\forall v\in\mathcal{V}; \label{eq:DSIC}
    \end{align}
    \textbf{Individual Rationality (IR)} if and only if
    \begin{align}
        u_i(v; \mathcal{M}) & \ge 0, & \forall i\in\mathcal{N},\forall v\in\mathcal{V}. \label{eq:IR}
    \end{align}
\end{definition}

\begin{definition}[Justified mechanism] \thlabel{def:justified}
    For an environment $\mathcal{E}=(\mathcal{N}, \mathcal{X}, \mathcal{V})$,
    a mechanism $\mathcal{M} = (\phi, \tau)$ is said to be \textbf{justified}
    if and only if it satisfies SE~\eqref{eq:SE}, DSIC~\eqref{eq:DSIC}, and IR~\eqref{eq:IR}.
\end{definition}
Feasibility for a justified mechanism to suppress the budget $B(v;\mathcal{M})$ in \thref{def:AMD} at most zero (weak budget balance; WBB),
or exactly equal to zero (strong budget balance; SBB),
is known with negative results even in a wide class of bilateral ($|\mathcal{N}| = 2$) environments
\cite{myerson1983efficient,Othman2009,Osogami2023}.
We demonstrate in Section~\ref{sec:proposed} how
the proposed justified mechanism is designed to
achieve the minimum budget
simultaneously for every type $v\in\mathcal{V}$,
with an arbitrarily specified option rule that satisfies SE~\eqref{eq:SE}.


%% file: vcg.tex
\section{Vickrey-Clarke-Groves (VCG) mechanism} \label{sec:vcg}

In the majority of mechanism design,
such an option rule $\phi$ that satisfies SE~\eqref{eq:SE} often comes with the following payment rule $\tau$~\eqref{eq:VCG} to form the celebrated VCG family in \thref{def:VCG}.
Then \thref{thm:VCG_DSIC} claims that any VCG mechanism satisfies DSIC~\eqref{eq:DSIC}.

\begin{definition}[VCG mechanism] \thlabel{def:VCG}
    A mechanism $\mathcal{M} = (\phi, \tau)$ is called the \textbf{Vickrey-Clarke-Groves (VCG) mechanism} if and only if the option rule $\phi$ satisfies SE~\eqref{eq:SE} and
    \begin{align}
        \tau_i(v) & = S(\phi(v); v) - v_i(\phi(v)) - h_i(v_{-i}) \nonumber                                                                               \\
                  & = \sum_{j\in\mathcal{N}\setminus\{i\}} v_j(\phi(v)) - h_i(v_{-i}), & \forall i\in\mathcal{N},\forall v\in\mathcal{V}, \label{eq:VCG}
    \end{align}
    with some functions $(h_i:\mathcal{V}_{-i}\rightarrow\mathbb{R})_{i\in\mathcal{N}}$.
\end{definition}

\begin{theorem}[Theorem 1.17 in \cite{Nisan2007}] \thlabel{thm:VCG_DSIC}
    The VCG mechanism satisfies DSIC~\eqref{eq:DSIC}.
\end{theorem}

The so-called Clarke pivot rule in \thref{def:VCG-Clarke} is one reasonable way to choose $h_i$s in \thref{def:VCG}.
Such defined a mechanism,
VCG-Clarke,
\emph{conditionally} guarantees IR~\eqref{eq:IR}.
It is common in the literature to implicitly assume environments where the maximum in~\eqref{eq:Clarke} exists,
but replacing the maximum with the supremum preserves equivalent qualities without the assumption.

\begin{definition}[VCG-Clarke] \thlabel{def:VCG-Clarke}
    The Clarke pivot rule specifies the payment rule~\eqref{eq:VCG} by
    \begin{align}
        h_i(v_{-i}) = h^c_i(v_{-i}) & \coloneqq \max_{X\in\mathcal{X}} \left(S(X;v) - v_i(X)\right) \nonumber                                                                                \\
                                    & = \max_{X\in\mathcal{X}} \sum_{j\in\mathcal{N}\setminus\{i\}} v_j(X),   & \forall v_{-i}\in\mathcal{V}_{-i},\forall i\in\mathcal{N}, \label{eq:Clarke}
    \end{align}
    and we call this particular implementation of the VCG mechanism as \textbf{VCG-Clarke}.
\end{definition}

\begin{lemma}[Lemma 1.20 in \cite{Nisan2007}] \thlabel{lem:VCG_IR}
    VCG-Clarke $\mathcal{M}^{c} = (\phi,\tau^{c})$ makes every agent pay to the broker:
    \begin{align}
        \tau^{c}_i(v) & \le 0, & \forall i\in\mathcal{N},\forall v\in\mathcal{V}. \label{eq:Clarke_comp}
    \end{align}
    It satisfies IR~\eqref{eq:IR} and thus justified if all possible types are non-negative i.e.,
    \begin{align}
        v_i(X) & \ge 0, & \forall v_i\in\mathcal{V}_i,\forall i\in\mathcal{N},\forall X\in\mathcal{X}. \label{eq:positive_type}
    \end{align}
\end{lemma}

The VCG-Clarke mechanism makes each agent $i\in\mathcal{N}$ pay to the broker the quantity
that expresses the negative effect that $i$'s participation has on the other agents $\mathcal{N}\setminus\{i\}$ (see \eqref{eq:externality} in the proof of \thref{lem:VCG_IR}).
Hence,
agents who have negative valuation on the option $\phi(v)$ result in negative utility;
IR may thus be violated in such cases.
In order to retain IR even when \eqref{eq:positive_type} does not hold,
we introduce an unconditionally justified VCG mechanism,
VCG-budget,
in \thref{def:VCG-budget}.
Moreover,
this mechanism achieves the minimum budget that can be obtained by any justified VCG mechanism,
as long as the option rule $\phi$ is fixed.

\begin{definition}[VCG-budget] \thlabel{def:VCG-budget}
    \textbf{VCG-budget} is defined as the VCG mechanism implemented by
    \begin{align}
        h_i(v_{-i}) = h^b_i(v_{-i})
         & \coloneqq  \min_{v_i'\in\mathcal{V}_i} S(\phi(v_i', v_{-i}); v_i', v_{-i})                                                                                 \nonumber                                                                                    \\
         & = \min_{v_i'\in\mathcal{V}_i}\max_{X\in\mathcal{X}} S(X; v_i',v_{-i}),                                                                                               & \forall v_{-i}\in\mathcal{V}_{-i},\forall i\in\mathcal{N}. \label{eq:VCG-budget}
    \end{align}
\end{definition}

\begin{lemma} \thlabel{lem:VCG-budget_IR}
    VCG-budget satisfies IR~\eqref{eq:IR};
    hence it is justified.
\end{lemma}

\begin{theorem} \thlabel{thm:VCG-budget_opt}
    VCG-budget $\mathcal{M}^{b} = (\phi, \tau^{b})$ satisfies,
    against any justified VCG mechanism $\mathcal{M} = (\phi, \tau)$,
    \begin{align}
        \tau^{b}_i(v)        & \le \tau_i(v),        & \forall i\in\mathcal{N},\forall v\in\mathcal{V}. \label{eq:vcg_tau_opt} \\
        B(v;\mathcal{M}^{b}) & \le B(v;\mathcal{M}), & \forall v\in\mathcal{V}. \label{eq:vcg_rev_opt}
    \end{align}
\end{theorem}

Proofs of \thref{lem:VCG-budget_IR} and \thref{thm:VCG-budget_opt} are straightforward and provided in Appendix~\ref{sec:proof}.

%% file: proposed.tex
\section{Proposed mechanism} \label{sec:proposed}

Beyond the VCG family,
the present study explores the entire universe of justified mechanisms,
out of which we identify and propose a budget-optimal mechanism $\mathcal{M}^* = (\phi^*, \tau^*)$ in Algorithm~\ref{alg:main}.
It consists of two functions that respectively represent its option rule $\phi^*$ and payment rule $\tau^*$:
\textproc{SelectOption},
which takes in a type list $v\in\mathcal{V}$ and returns an option $\phi^*(v)\in\mathcal{X}$,
and \textproc{ComputePayments},
which takes in a type list $v\in\mathcal{V}$ and returns a list of payments $\tau^*(v)\in\mathbb{R}^\mathcal{N}$.
The option rule $\phi^*$ is arbitrarily specified in advance as long as it satisfies SE~\eqref{eq:SE}.
It is verified in Section~\ref{subsec:correct} that the proposed mechanism is justified,
followed by detailed discussions in Section~\ref{subsec:optimal} about minimizing the budget,
with the computational complexity analyzed finally in~Section~\ref{subsec:complexity}.

\begin{algorithm}[t]
    \caption{The proposed mechanism built with an option rule $\phi^*$ that satisfies SE~\eqref{eq:SE}.}
    \begin{algorithmic}[1]
        \Function{SelectOption}{$v$} \Comment{Act as the option rule $\phi^*$.}
        \State \Return $\phi^*(v)$. 
        \EndFunction

        \Function{ComputePayments}{$v$} \Comment{Act as the payment rule $\tau^*$.}
        \For{$i\in\mathcal{N}$}
        \For{$v_i'\in\mathcal{V}_i$}
        \State $o(v_i')\gets\textproc{SelectOption}(v_i', v_{-i})$
        \Comment{Find the options for all possible change of $v_i$.}
        \EndFor
        \State $V \gets\mathcal{V}_i\cup\{\star\}$ \Comment{Define the set of vertices.}
        \State $E \gets V \times \mathcal{V}_i$ \Comment{Define the set of directed edges.}
        \For{$v_i'\in\mathcal{V}_i$}
        \State $c(\star, v_i') \gets v_i'(o(v_i'))$ \Comment{Set the weights of edges from $\star$.}
        \EndFor
        \For{$(v_i^{(1)},v_i^{(2)})\in \mathcal{V}_i^2$}
        \State $c(v_i^{(1)},v_i^{(2)}) \gets v_i^{(2)}(o(v_i^{(2)})) - v_i^{(2)}(o(v_i^{(1)}))$
        \Comment{Set the weights of remaining edges.}
        \EndFor
        \State $\tau^*_i(v) \gets -\Call{ShortestDistance}{\star, v_i; (V,E,c)}$
        \Comment{Compute the shortest distance.}
        \EndFor
        \State \Return $(\tau^*_i(v))_{i\in\mathcal{N}}$
        \EndFunction
    \end{algorithmic}
    \label{alg:main}
\end{algorithm}

In $\Call{ComputePayments}{v}$,
the payments $(\tau^*_i(v))_{i\in\mathcal{N}}$ is independently computed for each agent $i\in\mathcal{N}$.
First,
we build a weighted graph $G_i(v)$ that spans the vertex set $\mathcal{V}_i\cup\{\star\}$,
the agent's type domain $\mathcal{V}_i$ plus an auxiliary vertex $\star$.
Then for every $v_i'\in\mathcal{V}_{i}$,
a directed edge from $\star$ to $v_i'$ is equipped with a weight of $v'_i(\phi^*(v'_i, v_{-i}))$.
We also make $\mathcal{V}_i$,
all vertices except $\star$,
mutually connected by adding a directed edge for every ordered pair $(v_i^{(1)},v_i^{(2)})\in\mathcal{V}_i^2$ that weighs $v_i^{(2)}(\phi^*(v_i^{(2)},v_{-i})) - v_i^{(2)}(\phi^*(v_i^{(1)},v_{-i}))$.
Finally,
the oracle $\textproc{ShortestDistance}$ is called to exactly compute the shortest distance from the source node $\star$ to the destination $v_i$ in the graph $G_i(v)$.
The idea comes from the observation that the inequalities required by DSIC~\eqref{eq:DSIC} and IR~\eqref{eq:IR} are seen as a set of dual constraints for the shortest path problem (cf. Appendix~\ref{sec:duality}).
\thref{ex1} in Section~\ref{subsec:optimal} helps get a better sense of how the graph $G_i(v)$ is constructed and involved with the payment rule.
Further examples are provided in Appendix~\ref{sec:example}.

\subsection{Correctness} \label{subsec:correct}

The following lemma guarantees that \textproc{ShortestDistance} successfully computes the shortest distances in the graph $G_i(v)$.

\begin{lemma} \thlabel{lem:no-neg-cyc}
    For any $v\in\mathcal{V}$ and $i\in\mathcal{N}$,
    the weighted directed graph $G_i(v)$ has no negative directed cycle.
\end{lemma}

\begin{proof}
    Fix an arbitrary pair of $i\in\mathcal{N}$ and $v\in\mathcal{V}$
    for which we are going to show the claim holds.
    Note that no directed cycle contains the auxiliary vertex $\star$,
    since no edge comes out from it.
    Consider an arbitrary sequence $v_i^{(1)},\ldots,v_i^{(k)}$ in $\mathcal{V}_i$,
    and let $v_i^{(0)} \coloneqq v_i^{(k)}$.
    Let also $v^{(\ell)}\coloneqq (v_i^{(\ell)}, v_{-i}),\forall\ell\in\{0,1,\ldots,k\}$.
    By definition of \textproc{SelectOption},
    it follows that
    \begin{align}
         & v_i^{(\ell)}(\phi^*(v^{(\ell)})) + \sum_{j\in\mathcal{N}\setminus\{i\}} v_{j}(\phi^*(v^{(\ell)})) =   S(\phi^*(v^{(\ell)});v^{(\ell)})          \nonumber                                                              \\
         & \ge  S(\phi^*(v^{(\ell-1)});v^{(\ell)})                                                                       = v_i^{(\ell)}(\phi^*(v^{(\ell-1)})) + \sum_{j\in\mathcal{N}\setminus\{i\}} v_{j}(\phi^*(v^{(\ell-1)})),
         & \forall\ell\in\{1,\ldots,k\}. \label{eq:maxop}
    \end{align}
    Summing up \eqref{eq:maxop} over all $\ell\in\{1,\ldots,k\}$ yields
    \begin{align}
         & \sum_{\ell=1}^k (v_i^{(\ell)}(\phi^*(v^{(\ell)})) - v_i^{(\ell)}(\phi^*(v^{(\ell-1)})))      \nonumber \\
         & \ge \sum_{\ell=1}^k \left(\sum_{j\in\mathcal{N}\setminus\{i\}}  v_{j}(\phi^*(v^{(\ell-1)}))
        -\sum_{j\in\mathcal{N}\setminus\{i\}} v_{j}(\phi^*(v^{(\ell)})) \right) = 0, \label{eq:nonnegative-cycle}
    \end{align}
    where the last equality follows from $v^{(0)} = v^{(k)}$ by definition.
    Eq.~\eqref{eq:nonnegative-cycle} implies that the sum of the edge weights along any directed cycle in $G_i(v)$ is not negative.
\end{proof}

By virtue of \thref{lem:no-neg-cyc},
$-\tau^*_i(v)$ is successfully set to the shortest distance from $\star$ to $v_i$ in the graph $G_i(v)$.
Confirm that $G_i(v)$ does not depend on $i$'s own type $v_i$ by definition i.e.,
\begin{align}
    G_i(v) = G_i(v_i', v_{-i}), \quad \forall i\in\mathcal{N},\forall v\in\mathcal{V},\forall v_i'\in\mathcal{V}_i.
\end{align}
Hence,
$-\tau_i^*(v)$ and $-\tau_i^*(v_i',v_{-i})$ are obtained from the shortest paths on the same graph.
Therefore, the following inequalities hold by definition:
\begin{align}
    -\tau^*_i(v) & \le -\tau^*_i(v_i',v_{-i}) + v_i(\phi^*(v)) - v_i(\phi^*(v_i', v_{-i})),
                 & \forall i\in\mathcal{N},\forall v\in\mathcal{V},\forall v_i'\in\mathcal{V}_i,
\end{align}
which is equivalent to DSIC \eqref{eq:DSIC},
and
\begin{align}
    -\tau_i(v) & \le v_i(\phi^*(v)),
               & \forall i\in\mathcal{N},\forall v\in\mathcal{V},
\end{align}
which is equivalent to IR \eqref{eq:IR}.
Since $\phi^*$ satisfies SE \eqref{eq:SE} by definition,
we obtain \thref{thm:justified} below.

\begin{theorem} \thlabel{thm:justified}
    The proposed mechanism $\mathcal{M}^*$ is justified.
\end{theorem}

\subsection{Optimality} \label{subsec:optimal}

\thref{thm:maximal} claims that our mechanism $\mathcal{M}^*=(\phi^*,\tau^*)$ achieves the minimum budget among all justified mechanisms built with the same option rule~$\phi^*$.

\begin{theorem} \thlabel{thm:maximal}
    The proposed mechanism $\mathcal{M}^*$ satisfies,
    against any justified mechanism $\mathcal{M} = (\phi^*, \tau)$,
    \begin{align}
        \tau^*_i(v)        & \le \tau_i(v),        & \forall i\in\mathcal{N},\forall v\in\mathcal{V}, \label{eq:tau_maximized} \\
        B(v;\mathcal{M}^*) & \le B(v;\mathcal{M}), & \forall v\in\mathcal{V}. \label{eq:rev_maximized}
    \end{align}
\end{theorem}

\begin{proof}
    Fix an arbitrary pair of $i\in\mathcal{N}$ and $v\in\mathcal{V}$ for which we will show \eqref{eq:tau_maximized}.
    Now let a sequence of vertices $(\star, v_i^{(0)}, v_i^{(1)}, \ldots, v_i^{(k)} \coloneqq v_i)$  in the graph $G_i(v)$ be one of the shortest paths from the source $\star$ to the destination $v_i$.
    Again let $v^{(\ell)}\coloneqq (v_i^{(\ell)},v_{-i}),\forall\ell\in\{0,1,\ldots,k\}$.
    Since the vertex $v_i^{(0)}$ is visited right after the source $\star$ along the shortest path,
    we have
    \begin{align}
        -\tau^*_i(v^{(0)})                      = v_i^{(0)}(\phi^*(v^{(0)})).     \label{eq:thm2:1}
    \end{align}
    Similarly,
    since the vertex $v_i^{(\ell)}$ follows right after $v_i^{(\ell-1)}$ along the shortest path,
    we obtain
    \begin{align}
        -\tau^*_i(v^{(\ell)}) - (-\tau^*_i(v^{(\ell-1)})) & = v_i^{(\ell)}(\phi^*(v^{(\ell)})) - v_i^{(\ell)}(\phi^*(v^{(\ell-1)})), \quad\forall\ell\in\{1,\ldots,k\}. \label{eq:thm2:ell}
    \end{align}
    Then \eqref{eq:thm2:1} and \eqref{eq:thm2:ell} add up to
    \begin{align}
        -\tau^*_i(v^{(k)})
        = v_i^{(0)}(\phi^*(v^{(0)})) + \sum_{\ell=1}^{k} (v_i^{(\ell)}(\phi^*(v^{(\ell)})) - v_i^{(\ell)}(\phi^*(v^{(\ell-1)}))). \label{tau_proposed}
    \end{align}
    Next,
    let $\mathcal{M} = (\phi^*, \tau)$ be any justified mechanism.
    Then it follows from DSIC~\eqref{eq:DSIC} that
    \begin{align}
        -\tau_i(v^{(\ell)}) - (-\tau_i(v^{(\ell-1)})) & \le v_i^{(\ell)}(\phi^*(v^{(\ell)})) - v_i^{(\ell)}(\phi^*(v^{(\ell-1)})), \quad\forall \ell\in\{1,\ldots,k\},
    \end{align}
    and from IR~\eqref{eq:IR} that
    \begin{align}
        -\tau_i(v^{(0)}) & \le v_i^{(0)}(\phi^*(v^{(0)})),
    \end{align}
    which are tied together into
    \begin{align}
        -\tau_i(v^{(k)}) \le v_i^{(0)}(\phi^*(v^{(0)})) + \sum_{\ell=1}^{k} (v_i^{(\ell)}(\phi^*(v^{(\ell)})) - v_i^{(\ell)}(\phi^*(v^{(\ell-1)}))). \label{tau_general}
    \end{align}
    Finally, the desired inequality \eqref{eq:tau_maximized}
    follows from \eqref{tau_proposed} and \eqref{tau_general} since $v^{(k)} = v$,
    and then \eqref{eq:tau_maximized} adds up over $\mathcal{N}$ to \eqref{eq:rev_maximized}.
\end{proof}



When there is only one option rule that is SE, the proposed mechanism indeed minimizes the budget among all justified mechanisms, as is formally stated in the following corollary:
\begin{corollary} \thlabel{cor:rev}
    In ``proper'' environments,
    where the social welfare $S(X;v)$ is maximized by a unique option $X\in\mathcal{X}$ for each $v\in\mathcal{V}$,
    the proposed mechanism $\mathcal{M}^*$ achieves the minimum budget $B(v;\mathcal{M})$ that can be obtained by any justified mechanism $\mathcal{M}$.
    Formally,
    we claim the following:
    \begin{align}
        \left|\argmax_{X\in\mathcal{X}} S(X;v)\right| = 1,\ \forall v\in \mathcal{V}\quad
        \Rightarrow\quad B(v; \mathcal{M}^*) = \min_{\mathcal{M}:\,\mathrm{justified}} B(v; \mathcal{M}),\ \forall v\in \mathcal{V}.\label{eq:proper}
    \end{align}
\end{corollary}


When multiple option rules satisfy SE, the budget required by the proposed mechanism may vary depending on which option rule is used, even though it is minimized for any given option rule.

\begin{example}[Improper environment] \thlabel{ex1}
    This example provides an improper environment where the proposed mechanism may fail to globally minimize the budget for lack of the assumption in \thref{cor:rev}.
    Let us consider an environment with two agents $\mathcal{N} \coloneqq \{A, B\}$,
    three options $\mathcal{X} \coloneqq \{X_1, X_2, X_3\}$,
    and type domains $\mathcal{V}_A$ and $\mathcal{V}_B$ described in Table~\ref{table1}.

    We configure the option rule~$\phi^*$ such that
    \begin{align}
        \phi^*(v_{A}^{(1)}, v_{B}) \in \argmax_{X\in\mathcal{X}} (v_{A}^{(1)}(X) + v_{B}(X)) & = \{X_1\},                          \\
        \phi^*(v_{A}^{(2)}, v_{B}) \in \argmax_{X\in\mathcal{X}} (v_{A}^{(2)}(X) + v_{B}(X)) & = \{X_2, X_3\}. \label{eq:ex1:phi2}
    \end{align}
    If $X_2$ is selected for $\phi^*(v_{A}^{(2)}, v_{B})$ in~\eqref{eq:ex1:phi2},
    then the payments on the types~$v\coloneqq(v_{A}^{(1)}, v_{B})$ would be
    \begin{align}
        \tau^*_A(v) = \alpha, \quad
        \tau^*_B(v) =  0, \label{eq:ex1:case1}
    \end{align}
    which are computed as shortest distances in the graph shown in Figure~\ref{fig:ex1}\ref{sub@fig:ex1:case1}.
    Otherwise,
    if $X_3$ is selected as $\phi^*(v_{A}^{(2)}, v_{B})$ in~\eqref{eq:ex1:phi2},
    then the graph in Figure~\ref{fig:ex1}\ref{sub@fig:ex1:case2} would tell
    \begin{align}
        \tau^*_A(v) = -\alpha, \quad
        \tau^*_B(v) = 0. \label{eq:ex1:case2}
    \end{align}
    Therefore,
    given the types~$v$,
    the proposed mechanism might end up with the budget $B(v;\mathcal{M}^*)=\tau^*_A(v)+\tau^*_B(v)=\alpha$ due to~\eqref{eq:ex1:case1},
    which is higher than otherwise possible budget of $-\alpha$ in the case of~\eqref{eq:ex1:case2}.
    Under this environment,
    expense to guarantee DSIC~\eqref{eq:DSIC} is intuitively attributed to difference in valuation between agent A's two possible types $\mathcal{V}_{A} = \{v_{A}^{(1)},v_{A}^{(2)}\}$,
    which could be larger with the option~$X_2$ than with~$X_3$.

\end{example}

\begin{figure}[tb]
    \begin{minipage}{0.36\linewidth}
        \tblcaption{An improper environment parameterized with $\alpha > 0$.
            Agent~$A$ can have one of two types $\mathcal{V}_A \coloneqq \{v_{A}^{(1)}, v_{A}^{(2)}\}$,
            while agent~$B$ has only one possible type $\mathcal{V}_B \coloneqq \{v_{B}\}$.
        }
        \centering
        \begin{tabular}{lccc}                                                                       \\\toprule
            \multirow{2}{*}{Type} & \multicolumn{3}{c}{Option}                           \\ \cmidrule{2-4}
                                  & $X_1$                      & $X_2$      & $X_3$      \\ \midrule
            $v_{A}^{(1)}$         & $\alpha$                   & $0$        & $0$        \\ \addlinespace[1.5pt]
            $v_{A}^{(2)}$         & $-3\alpha$                 & $-2\alpha$ & $0$        \\ \addlinespace[1.5pt]
            $v_{B}$               & $0$                        & $0$        & $-2\alpha$ \\ \bottomrule
        \end{tabular}
        \label{table1}
    \end{minipage}
    \hfill
    \begin{minipage}{0.58\linewidth}
        \centering
        \begin{subfigure}{0.49\textwidth}
            \centering
            \begin{adjustbox}{scale=0.7}
                \begin{tikzpicture}[->,>=stealth',shorten >=1pt,auto,node distance=1.5cm,thick,
                        vertex/.style={circle,draw,font=\sffamily\bfseries}
                    ]
                    \node[vertex](s){\Large{$\star$}};
                    \node[vertex](a1)[above left=0.7cm and 1.2cm of s]{$v_{A}^{(1)}$};
                    \node[vertex](a2)[below left=0.7cm and 1.2cm of s]{$v_{A}^{(2)}$};
                    \node[vertex](b)[right=1.2cm of s]{$v_{B}$};

                    \path[every node/.style={font=\sffamily\small}]
                    (s) edge node[above]{$0$} (b)
                    (s) edge node[above right]{$\alpha$} (a1)
                    (s) edge node[below right]{$-2\alpha$} (a2)
                    (a1) edge[bend left] node[right]{$\alpha$} (a2)
                    (a2) edge[bend left] node[left]{$\alpha$} (a1);

                    \node[below left=0.25cm and 0.4cm of a2] (ga_lb) {};
                    \node[above right=0.25cm and 1.25cm of a1] (ga_ru) {};
                    \draw[gray,dotted] (ga_lb) rectangle (ga_ru);
                    \node[below left=0.05cm and 0.05cm of ga_ru,gray] {$G_A$};

                    \node[right=2.6cm of ga_lb] (gb_lb) {};
                    \node[right=2.25cm of ga_ru] (gb_ru) {};
                    \draw[gray,dotted] (gb_lb) rectangle (gb_ru);
                    \node[below left=0.05cm and 0.05cm of gb_ru,gray] {$G_B$};
                \end{tikzpicture}
            \end{adjustbox}
            \caption{$\phi^*(v_{A}^{(2)}, v_{B}) = X_2$.}
            \label{fig:ex1:case1}
        \end{subfigure}
        \begin{subfigure}{0.49\textwidth}
            \centering
            \begin{adjustbox}{scale=0.7}
                \begin{tikzpicture}[->,>=stealth',shorten >=1pt,auto,node distance=1.5cm,thick,
                        vertex/.style={circle,draw,font=\sffamily\bfseries}
                    ]
                    \node[vertex](s){\Large{$\star$}};
                    \node[vertex](a1)[above left=0.7cm and 1.2cm of s]{$v_{A}^{(1)}$};
                    \node[vertex](a2)[below left=0.7cm and 1.2cm of s]{$v_{A}^{(2)}$};
                    \node[vertex](b)[right=1.2cm of s]{$v_{B}$};

                    \path[every node/.style={font=\sffamily\small}]
                    (s) edge node[above]{$0$} (b)
                    (s) edge node[above right]{$\alpha$} (a1)
                    (s) edge node[below right]{$0$} (a2)
                    (a1) edge[bend left] node[right]{$3\alpha$} (a2)
                    (a2) edge[bend left] node[left]{$\alpha$} (a1);

                    \node[below left=0.25cm and 0.4cm of a2] (ga_lb) {};
                    \node[above right=0.25cm and 1.25cm of a1] (ga_ru) {};
                    \draw[gray,dotted] (ga_lb) rectangle (ga_ru);
                    \node[below left=0.05cm and 0.05cm of ga_ru,gray] {$G_A$};

                    \node[right=2.6cm of ga_lb] (gb_lb) {};
                    \node[right=2.25cm of ga_ru] (gb_ru) {};
                    \draw[gray,dotted] (gb_lb) rectangle (gb_ru);
                    \node[below left=0.05cm and 0.05cm of gb_ru,gray] {$G_B$};
                \end{tikzpicture}
            \end{adjustbox}
            \caption{$\phi^*(v_{A}^{(2)}, v_{B}) = X_3$.}
            \label{fig:ex1:case2}
        \end{subfigure}

        \caption{The weighted graph $G_A(v)$ and $G_B(v)$ for the type~$v = (v_{A}^{(1)},v_B)$,
            combined at the auxiliary vertex~$\star$,
            for the two possibilities of (\ref{sub@fig:ex1:case1}) and (\ref{sub@fig:ex1:case2}).}
        \label{fig:ex1}
    \end{minipage}
\end{figure}
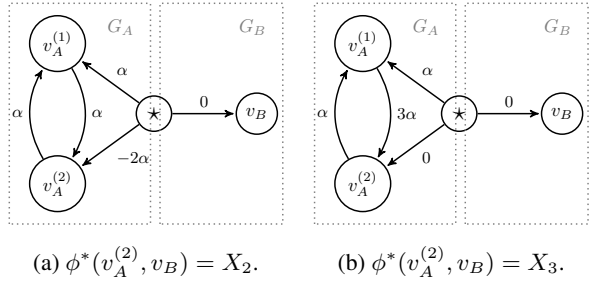

For such improper environments,
the following corollary, which we prove in Appendix~\ref{sec:proof}, suggests that we may search for an option rule among the ones that satisfy SE \eqref{eq:SE} that optimizes some aggregated metric based on the budgets $B(v;\mathcal{M})$ across types $v\in\mathcal{V}$:

\begin{corollary} \thlabel{cor:nonproper}
    Let $\Phi^*$ be the set of all option rules that satisfies SE \eqref{eq:SE}.
    Let $\mathcal{M}^*[\phi^*]$ be the proposed mechanism (Algorithm~\ref{alg:main}) built with an option rule $\phi^*\in\Phi$.
    Then, for any non-decreasing function $f:\mathbb{R}^{\mathcal{V}}\to\mathbb{R}$, we have
    \begin{align}
        \min_{\phi^*\in\Phi^*} f((B(v;\mathcal{M}^*[\phi^*]))_{v\in\mathcal{V}})
         & = \min_{\mathcal{M}:\,\mathrm{justified}} f((B(v;\mathcal{M}))_{v\in\mathcal{V}}).
        \label{eq:nonproper}
    \end{align}
\end{corollary}

For example,
one may let $f((B(v;\mathcal{M}))_{v\in\mathcal{V}}) = \mathbb{E}[B(v;\mathcal{V})]$,
where $\mathbb{E}$ denotes the expectation with respect to an arbitrarily assumed probability distribution over $\mathcal{V}$,
and choose an option rule $\phi^*$ that allows the mechanism $\mathcal{M}^*[\phi^*]$ to minimize the expected budget.

\input{complexity.tex}

%% file: complexity.tex
\subsection{Computational complexity} \label{subsec:complexity}

The running time of the proposed mechanism can depend on the size of the environment as well as the running time of the subroutines such as \textproc{SelectOption} and \textproc{ShortestDistance}.  The following theorem formally characterizes the computational complexity of the proposed mechanism and provides a way to reduce the complexity when only a limited number of options can be selected.

\begin{theorem} \thlabel{thm:time}
    \textproc{ComputePayments} runs in $\displaystyle O\left(\sum_{i\in\mathcal{N}}(|\mathcal{V}_i|T_{opt} + |\mathcal{V}_i|^2 T_{eval} + \mathrm{SP}(|\mathcal{V}_i|, |\mathcal{V}_i|^2))\right)$ time,
    where $T_{opt}$ denotes the worst time complexity of \textproc{SelectOption},
    $T_{eval}$ denotes that of evaluating any single value $v_i(X),\forall v_i\in\mathcal{V}_i,\forall i\in\mathcal{N},\forall X\in\mathcal{X}$,
    and $\mathrm{SP}(n,m)$ denotes that of the oracle \textproc{ShortestDistance} given a graph with $n$ vertices and $m$ edges.
    Furthermore,
    a modified implementation reduces the complexity to $\displaystyle O\left(\sum_{i\in\mathcal{N}}(|\mathcal{V}_i|T_{opt} + n_i|\mathcal{V}_i|T_{eval} + \mathrm{SP}(n_i, n_i^2))\right)$,
    with $n_i \coloneqq \displaystyle\max_{v_{-i}\in\mathcal{V}_{-i}}|\{\phi^*(v_i',v_{-i})\mid v_i'\in\mathcal{V}_i\}|, \forall i\in\mathcal{N}$.
\end{theorem}

\begin{proof}
    For each $i\in\mathcal{N}$,
    the first term denotes the complexity of \textproc{SelectOption} called $|\mathcal{V}_i|$ times,
    followed by the complexity of computing $|\mathcal{V}_i|^2 + |\mathcal{V}_i|$ edge weights,
    and finally by the complexity of the oracle \textproc{ShortestDistance} given the graph $G_i(v)$.

    The improved complexity is obtained by summarizing $G_i(v)$ into a smaller graph $\tilde{G}_i(v)$ of at most $n_i+1$ vertices,
    $\mathcal{X}_{i}(v_{-i})\cup\{\tilde{\star}\}$,
    where we define
    \begin{align}
        \mathcal{X}_{i}(v_{-i}) & \coloneqq \{\phi^*(v_i',v_{-i}) \mid v_i'\in\mathcal{V}_i\}\subseteq\mathcal{X} , & \forall i\in\mathcal{N},\forall v\in\mathcal{V}.
    \end{align}
    The graph $\tilde{G}_i(v)$ is equipped with directed edges for all pairs in $(\mathcal{X}_{i}(v_{-i})\cup\{\tilde{\star}\})\times \mathcal{X}_{i}(v_{-i})$,
    whose weights $\tilde{c}$ are defined as follows:
    \begin{align}
        \tilde{c}(\tilde{\star}, \Phi) & \coloneqq \min\{v_i'(\Phi) \mid v_i'\in\mathcal{V}_i,\ \phi^*(v_i',v_{-i}) = \Phi \},                 & \forall \Phi\in\mathcal{X}_{i}(v_{-i}),              \\
        \tilde{c}(\Phi_1, \Phi_2)      & \coloneqq \min\{v_i'(\Phi_2)-v_i'(\Phi_1) \mid v_i'\in\mathcal{V}_i,\ \phi^*(v_i',v_{-i}) = \Phi_2\}, & \forall (\Phi_1,\Phi_2)\in\mathcal{X}_{i}(v_{-i})^2.
    \end{align}
    From another viewpoint,
    the graph $\tilde{G}_i(v)$ is obtained from $G_i(v)$ by contracting vertices $v_i'\in\mathcal{V}_{i}$ with the same option $\phi^*(v_i', v_{-i})\in\mathcal{X}_{i}(v_{-i})$
    into a single vertex labeled with it,
    while resulting parallel edges in $G_i(v)$ are replaced by an aggregated single edge in $\tilde{G}_i(v)$ with the minimum weight among those of the original edges.
    The desired shortest distance~$-\tau^*_i(v)$ in~$G_i(v)$ from~$\star$ to~$v_i$ remains as the shortest distance in $\tilde{G}_i(v)$ from $\tilde{\star}$ to $\phi^*(v)$,
    because in the original graph $G_i(v)$,
    vertices $v_i'\in\mathcal{V}_i$ with a common $\phi^*(v_i',v_{-i})$
    are mutually and bidirectionally connected by zero-weight edges,
    and thus share an equal shortest distance from $\star$.
\end{proof}

Note $n_i\le \min\{|\mathcal{V}_i|,|\mathcal{X}|\},\forall i\in\mathcal{N}$ and that $\mathrm{SP}(n, m) = O(nm)$ is guaranteed as the current best bound by the Bellman-Ford algorithm~\cite{Madkour2017}.
The present mechanism is thus executed in polynomial time with respect to the size of the type domains $|\mathcal{V}_i|$.
However,
if the number of options $|\mathcal{X}|$ is large compared with $|\mathcal{V}_i|$,
or even infinite,
then $T_{opt}$ would potentially become an intractable bottleneck.
In such cases,
one may preferably leverage some analytical properties about types (e.g., concavity) for specific environments.
The above proof of \thref{thm:time} also implies the following space complexity.

\begin{corollary} \thlabel{cor:space}
    \textproc{ComputePayments} requires $\displaystyle O\left(S_{opt} + S_{eval} + \max_{i\in\mathcal{N}}|\mathcal{V}_i|^2\right)$ space,
    where $S_{opt}$ denotes the worst space complexity of \textproc{SelectOption} and
    $S_{eval}$ denotes that of evaluating any single value $v_i(X),\forall v_i\in\mathcal{V}_i,\forall i\in\mathcal{N},\forall X\in\mathcal{X}$.
    Furthermore,
    the complexity can be reduced to $\displaystyle O\left(S_{opt} + S_{eval} + \max_{i\in\mathcal{N}}n_i^2 \right)$,
    with $n_i \coloneqq \displaystyle\max_{v_{-i}\in\mathcal{V}_{-i}}|\{\phi(v_i',v_{-i})\mid v_i'\in\mathcal{V}_i\}|, \forall i\in\mathcal{N}$.
\end{corollary}

%% file: experiments.tex
\section{Numerical experiments} \label{sec:exp}

Here we empirically investigate the quantitative performance of the proposed mechanism in terms of minimizing the budget.
\thref{thm:VCG-budget_opt} in Section~\ref{sec:vcg} confirms that VCG-budget (\thref{def:VCG-budget}) is optimal among justified VCG mechanisms,
which motivates us to compare the proposed mechanism with VCG-budget as a representative of all justified VCG mechanisms.
The experiments are thus focused on whether and how often the proposed mechanism achieves strictly lower budget than VCG-budget.


Every instance is generated as a pair of an environment $\mathcal{E} = (\mathcal{N},\mathcal{X},\mathcal{V})$ and a list of types $v\in\mathcal{V}$ for which the mechanisms are executed.
The size of the environment $\mathcal{E}$ is parametrized by the number of agents $|\mathcal{N}|\in\{1,2,\ldots,32\}$,
the number of options $|\mathcal{X}|\in\{1,2,\ldots,256\}$,
and the size of type domains $|\mathcal{V}_i|\in\{1,2,\ldots,16\},\forall i\in\mathcal{N}$,
each of which is either specified or randomly drawn when the instance is generated.
Then each value $v_i(X)$ is drawn independently and identically from the discrete uniform distribution over the set of consecutive integers from $-100$ to $100$.
Now that the environment $\mathcal{E}$ is configured, the list of types $v$ is finally chosen from the uniform distribution over $\mathcal{V} = \prod_{i\in\mathcal{N}}\mathcal{V}_i$.
Both VCG-budget and the proposed mechanism use a common option rule for each environment.
It is confirmed that either mechanism runs within one second at worst per instance on a laptop with no GPU,
a single Intel Core i7-11850H @ 2.50GHz,
and 64GB RAM.

Figure~\ref{fig:exp:hist} shows the histogram of the difference in the budget required by the proposed mechanism relative to VCG-budget.  Here, we only show the results with $|\mathcal{N}|=16$ (see Appendix~\ref{sec:more_exp} for other results).  The right-most bar corresponds to the cases where the VCG-budget is also optimal (difference to the proposed mechanism is zero).  It can thus be observed that the proposed mechanism requires strictly lower budget for 88.3\% of the cases in this setting.  While the exact frequency varies depending on the settings, we find that the proposed mechanism generally achieves strictly lower budget than VCG-budget for a large fraction of the instances (see also Appendix~\ref{sec:more_exp}).

\begin{figure}[tb]
    \centering
    \begin{subfigure}{0.245\linewidth}
        \includegraphics[scale=0.5]{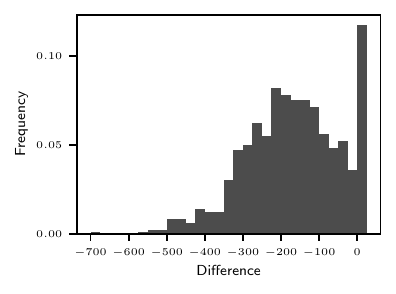}
        \caption{$|\mathcal{N}|=16$}
        \label{fig:exp:hist}
    \end{subfigure}
    \begin{subfigure}{0.245\linewidth}
        \includegraphics[scale=0.5]{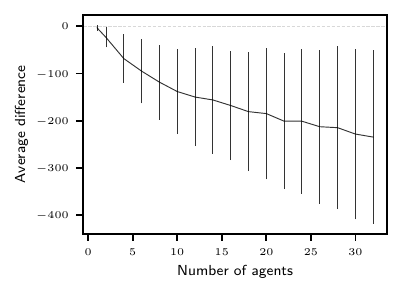}
        \caption{$1\le|\mathcal{N}|\le32$}
        \label{fig:exp:plot_vs_n}
    \end{subfigure}
    \begin{subfigure}{0.245\linewidth}
        \includegraphics[scale=0.5]{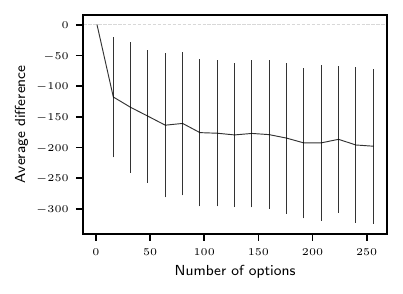}
        \caption{$1\le|\mathcal{X}|\le256$}
        \label{fig:exp:plot_vs_m}
    \end{subfigure}
    \begin{subfigure}{0.245\linewidth}
        \includegraphics[scale=0.5]{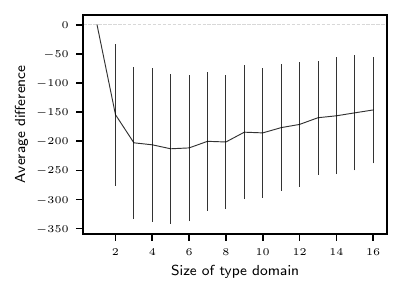}
        \caption{$1\le|\mathcal{V}_i|\le16$}
        \label{fig:exp:plot_vs_d}
    \end{subfigure}
    \caption{
      Difference in budget required by the proposed mechanism relative to VCG-budget, shown as
      a histogram in (\ref{sub@fig:exp:hist})
      as well as the average against the number of agents $|\mathcal{N}|$ in (\ref{sub@fig:exp:plot_vs_n}), 
      the number of options $|\mathcal{X}|$ in (\ref{sub@fig:exp:plot_vs_m}), and      
      the size of type domains $|\mathcal{V}_i|$ in (\ref{sub@fig:exp:plot_vs_d}),
      where $|\mathcal{N}|=16$ is fixed except in (\ref{sub@fig:exp:plot_vs_n}).
    }
    \label{fig:exp}
\end{figure}

Figure~\ref{fig:exp:plot_vs_n}-\ref{fig:exp:plot_vs_d} shows the average difference in the budget when we vary the number of agents $|\mathcal{N}|$ in (\ref{sub@fig:exp:plot_vs_n}), the number of options $|\mathcal{X}|$ in (\ref{sub@fig:exp:plot_vs_m}), and the size of type domain $|\mathcal{V}_i|$ in (\ref{sub@fig:exp:plot_vs_d}).  Each data point is the average over 1,000 random instances, and the error bars show standard deviation.  Overall, the proposed mechanism requires significantly lower budget than VCG-budget for all cases except when $|\mathcal{N}|$, $|\mathcal{X}|$, or $|\mathcal{V}_i|$ is close to one.  In this setting, the relative benefit of the proposed mechanism tends to increase with $|\mathcal{N}|$ and $|\mathcal{X}|$, but the proposed mechanism gives the largest improvement at intermediate values of $|\mathcal{V}_i|$.  See Appendix~\ref{sec:more_exp} for the results with other settings.

%% file: conclusion.tex
\section{Conclusion}

The present study has identified the budget-optimal mechanism among all the mechanisms that satisfies SE, DSIC, and IR in a constructive way coupled with theoretical evaluations.
Moreover,
through numerical experiments,
the proposed mechanism has proved itself beyond the VCG family and demonstrated strictly lower budget than any VCG mechanism for a majority of instances.
However,
there are limitations to the type domains to which this study can be applied, and future studies are open around the extension of our approach to continuous and/or unknown type domains,
where it would be impossible to enumerate the set of inequalities required by DSIC~\eqref{eq:DSIC} and IR~\eqref{eq:IR}.

As we discuss in Appendix~\ref{sec:ama}, the proposed mechanism can be extended to work with the option rule that is an affine maximizer instead of a (unweighted) social-welfare maximizer.  This extended mechanism can minimize the budget among all the mechanisms that maximize the weighted sum of the agents' valuations \emph{for given weights} while satisfying DSIC and IR.  It is an interesting direction of future work to leverage this extended mechanism for the purpose of maximizing the revenue by optimizing the weights similar to the affine maximizer auction.

%% file: proof.tex
\section{Proofs}
\label{sec:proof}

Here, we provide full proofs for the theorems, lemmas, and corollaries whose proofs are omitted in the body of the paper.

\begin{proof}[Proof of \thref{thm:VCG_DSIC}]
    Let $\mathcal{M} = (\phi, \tau)$ be the VCG mechanism implemented with any $(h_i)_{i\in\mathcal{N}}$.
    Then
    \begin{align}
        u_i(v;\mathcal{M}) & = v_i(\phi(v)) + \tau_i(v) \nonumber                                                                          \\
                           & = S(\phi(v);v) - h_i(v_{-i})  \nonumber                                                                       \\
                           & \ge S(\phi(v_i',v_{-i}); v) - h_i(v_{-i}) \nonumber                                                           \\
                           & = v_i(\phi(v_i',v_{-i})) + \sum_{j\in\mathcal{N}\setminus\{i\}}v_j(\phi(v_i',v_{-i})) - h_i(v_{-i}) \nonumber \\
                           & = v_i(\phi(v_i',v_{-i})) + \tau_i(\phi(v_i',v_{-i})),
                           & \forall v_i\in\mathcal{V}_i,\forall i\in\mathcal{N},\forall v\in\mathcal{V},
    \end{align}
    completes the proof.
\end{proof}

\begin{proof}[Proof of \thref{lem:VCG_IR}]
    The first claim~\eqref{eq:Clarke_comp} follows by definition:
    \begin{align}
        \tau^c_i(v) & = \sum_{j\in\mathcal{N}\setminus\{i\}} v_j(\phi(v)) - \max_{X\in\mathcal{X}} \sum_{j\in\mathcal{N}\setminus\{i\}} v_j(X) \le 0, & \forall i\in\mathcal{N},\forall v\in\mathcal{V}.\label{eq:externality}
    \end{align}
    It is also observed that
    \begin{align}
        u_i(v;\mathcal{M}) & = v_i(\phi(v)) + \tau^c_i(v) \nonumber                                                                                      \\
                           & = S(\phi(v);v) - h^c_i(v_{-i}) \nonumber                                                                                    \\
                           & = \max_{X\in\mathcal{X}} \sum_{j\in\mathcal{N}} v_j(X) - \max_{X\in\mathcal{X}} \sum_{j\in\mathcal{N}\setminus\{i\}} v_j(X)
        \ge \inf_{X\in\mathcal{X}} v_i(X),
                           & \forall i\in\mathcal{N},v\in\mathcal{V},
    \end{align}
    which yields IR~\eqref{eq:IR} when \eqref{eq:positive_type} holds.
\end{proof}

\begin{proof}[Proof of \thref{lem:VCG-budget_IR}]
    VCG-budget $\mathcal{M}^{b} = (\phi, \tau^{b})$ satisfies the following:
    \begin{align}
        u_i(v;\mathcal{M}^b) & = v_i(\phi(v)) + \tau^b_i(v) \nonumber                                           \\
                             & = S(\phi(v);v) - \min_{v_i'\in\mathcal{V}_i} S(\phi(v_i', v_{-i}); v_i', v_{-i})
        \ge 0,               & \forall i\in\mathcal{N},\forall v\in\mathcal{V},
    \end{align}
    which directly gives the claim.
\end{proof}

\begin{proof}[Proof of \thref{thm:VCG-budget_opt}]
    Let $\mathcal{M} = (\phi, \tau)$ be any justified VCG mechanism implemented with $(h_i)_{i\in\mathcal{N}}$.
    Then IR gives
    \begin{align}
        0 \le u_i(v_i',v_{-i}; \mathcal{M})
         & = v_i'(\phi(v_i',v_{-i})) + \tau_i(v_i',v_{-i}) \nonumber                                                                                            \\
         & = S(\phi(v_i',v_{-i});v_i',v_{-i}) - h_i(v_{-i}) ,        & \forall v_i'\in\mathcal{V}_i, \forall v_{-i}\in\mathcal{V}_{-i},\forall i\in\mathcal{N},
    \end{align}
    or equivalently summarized with \eqref{eq:VCG-budget} into
    \begin{align}
        0 & \le h^b_i(v_{-i}) - h_i(v_{-i}), & \forall v_{-i}\in\mathcal{V}_{-i},\forall i\in\mathcal{N},
    \end{align}
    which leads to \eqref{eq:vcg_tau_opt} and in particular \eqref{eq:vcg_rev_opt}.
\end{proof}

\begin{proof}[Proof of \thref{cor:nonproper}]
    \thref{thm:maximal} claims
    \begin{align}
        B(v;\mathcal{M}^*[\phi^*])
         & = \min_{(\phi^*,\tau):\,\mathrm{justified}} B(v;(\phi^*,\tau)), & \forall v\in\mathcal{V},\forall \phi^*\in\Phi^*.
    \end{align}
    Since $f$ is non-decreasing, the previous equality implies
    \begin{align}
        f((B(v;\mathcal{M}^*[\phi^*]))_{v\in\mathcal{V}})
         & =
        f\left(\left(\min_{(\phi^*,\tau):\,\mathrm{justified}} B(v;(\phi^*,\tau))\right)_{v\in\mathcal{V}}\right)              \\
         & \le \min_{(\phi^*,\tau):\,\mathrm{justified}} f((B(v;(\phi^*,\tau)))_{v\in\mathcal{V}}), & \forall \phi^*\in\Phi^*.
    \end{align}
    By taking the minimum of each side over $\Phi^*$, we obtain
    \begin{align}
        \min_{\phi^*\in\Phi^*} f((B(v;\mathcal{M}^*[\phi^*]))_{v\in\mathcal{V}})
         & \le \min_{\phi^*\in\Phi^*} \min_{(\phi^*,\tau):\,\mathrm{justified}} f((B(v;(\phi^*,\tau)))_{v\in\mathcal{V}}) \\
         & = \min_{\mathcal{M}:\,\mathrm{justified}} f((B(v;\mathcal{M}))_{v\in\mathcal{V}}). \label{eq:le}
    \end{align}
    Since $\{\mathcal{M}^*[\phi^*]: \phi^*\in\Phi^*\} \subseteq \{\mathcal{M}: \mathrm{justified}
        \}$, we also have
    \begin{align}
        \min_{\phi^*\in\Phi^*} f((B(v;\mathcal{M}^*[\phi^*]))_{v\in\mathcal{V}})
         & \ge \min_{\mathcal{M}:\,\mathrm{justified}} f((B(v;\mathcal{M}))_{v\in\mathcal{V}}),
    \end{align}
    which together with \eqref{eq:le} establishes the corollary.
\end{proof}

%% file: examples.tex
\section{Supplementary examples} \label{sec:example}

Examples in this section hopefully provide a clearer view of how the model and the proposed mechanism are applied specifically.

\begin{example}[Continuous options] \thlabel{ex:continuous}
    Let us consider, in a one-dimensional world,
    that the broker plans to host an onsite event and invite limited guests there.
    Suppose that a guest who lives at the location $b\in\mathbb{R}$ incurs a quadratic (physical and/or monetary) cost of $a(x - b)^2 + c$ for a round trip to the venue at $x\in\mathbb{R}$ by some transportation characterized by the constants $a\in\mathbb{R}_{>0}$ and $c\in\mathbb{R}$.
    The broker should locate the venue so that the total cost of guests is minimized,
    while he also wants to expense to the guests as little travel allowance as possible.

    This situation can be formulated as an environment $\mathcal{E} = (\mathcal{N},\mathcal{X},\mathcal{V})$ that consists of the guests~$\mathcal{N}$,
    a geographic range (closed interval) available for the venue~$\mathcal{X} \coloneqq [x_{min},x_{max}]\subseteq\mathbb{R}$,
    and type domains $\mathcal{V}_i \coloneqq \left\{v^{(k)}_i:\mathcal{X}\ni x\mapsto -a^{(k)}_i\left(x - b^{(k)}_i\right)^2 - c^{(k)}_i \in\mathbb{R} \mid k\in\{1,\ldots,d_i\}\right\}$ of size $d_i$ parameterized with $(a^{(k)}_i, b^{(k)}_i, c^{(k)}_i)_{k=1}^{d_i}\in(\mathbb{R}_{>0}\times\mathbb{R}\times\mathbb{R})^{d_i}$ for each $i\in\mathcal{N}$.
    Each guest $i\in\mathcal{N}$ reports his address and transportation as a type $v_i\in\mathcal{V}_i$ and receive allowance $\tau_i(v)$ from the broker,
    who does not want the guest to gain by false declaration.
    Given any types $v = (v^{(k_i)}_i)_{i\in\mathcal{N}}\in\mathcal{V}$,
    the social welfare is equal to the negative total cost of agents:
    \begin{align}
        S(x; v) = \sum_{i\in\mathcal{N}} \left(-a^{(k_i)}_i\left(x - b^{(k_i)}_i\right)^2 - c^{(k_i)}_i\right),
    \end{align}
    which is maximized by the following unique option:
    \begin{align}
        x = \phi^*(v) \coloneqq \max\left\{x_{min}, \min\left\{x_{max}, \dfrac{\sum_{i\in\mathcal{N}} a_i^{(k_i)}b_i^{(k_i)}}{ \sum_{i\in\mathcal{N}}a_i^{(k_i)}}\right\}\right\}.\label{eq:opt_location}
    \end{align}
    While the broker is obliged to maximize the social welfare (SE),
    he also has to incentivize agents to make truthful reports (DSIC)
    and compensate each agent $i$ with $\tau_i(v)$ at least $i$'s travel cost (IR),
    hopefully on the minimum possible budget $\sum_{i\in\mathcal{N}}\tau_i(v)$.
    The proposed mechanism perfectly succeeds along with it in this proper environment owing to \thref{cor:rev}.
\end{example}

\begin{example}[Vickery auction] \thlabel{ex:vickrey}
    Here we consider an auction for a single indivisible item,
    for which each agent makes a bid out of finite biddable prices.
    Formally, consider an environment with $n$ agents $\mathcal{N} \coloneqq \{1,\ldots,n\}$
    and the same number of options $\mathcal{X} \coloneqq \{X_i\mid i\in\mathcal{N}\}$,
    each of which $X_i\in\mathcal{X}$ represents the allocation where the item goes to the agent $i\in\mathcal{N}$.
    Each agent $i\in\mathcal{N}$ has a type domain $\mathcal{V}_i\coloneqq\{v_i^{(1)},\ldots,v_i^{(d)}\}$ of size $d$
    such that
    \begin{align}
        v_i^{(k)}(X_j) = \begin{cases}
                             p_k & \text{if } j=i,   \\
                             0   & \text{otherwise,}
                         \end{cases} \label{ex2_price},
         &  & \forall k\in\{1,\ldots,d\},\forall j\in\mathcal{N},\forall i\in\mathcal{N},
    \end{align}
    where $p_1 > p_2 > \cdots > p_d > 0$ are all biddable prices regardless of agent.

    Although it is an improper setting since multiple options maximize the social welfare if more than one agent make the highest bids,
    we do not lose generality by defining the option rule $\phi^*$ as follows.
    Given an arbitrary list of types $v\in\mathcal{V}$,
    let us define $\phi^*(v)$ as $X_i\in\argmax_{X\in\mathcal{X}} \sum_{i\in\mathcal{N}} v_i(X)$ with the minimum $i\in\mathcal{N}$.
    In other words,
    tie-breaking among agents who bid the highest is won by one with the smallest index.

    Let $v = (v_{i_{\ell}}^{(k_{\ell})})_{\ell\in\mathcal{N}}\in\mathcal{V}$ be an arbitrary list of types,
    where $(i_\ell)_{\ell\in\mathcal{N}}$ is a permutation of $\mathcal{N}$ that is sorted in the lexicographical order on pairs $(k_{\ell}, i_{\ell})$.
    The agent $i_1$ makes the highest bid of $p_{k_1}$,
    followed by $i_2$,
    who makes an equal or less bid of $p_{k_2}$.
    It follows from \eqref{ex2_price} and the definition of $\phi^*$ that
    $\phi^*(v) = X_{i_1}$,
    which closes the auction with only $i_1$ making a successful bid of $p_{k_1}$.

    Now we focus on the payment~$-\tau^*_{i_1}(v)$ from the winner~$i_1$.
    It is equal to the shortest distance from the auxiliary vertex $\star$ to the vertex $v_{i_1}^{(k_1)}$ in the weighted graph~$G_{i_1}(v)$.
    If $i_1 < i_2$,
    the agent $i_1$ would still win if they bid at least $p_{k_2}$;
    otherwise, the agent~$i_2$ would win against~$i_1$ i.e.,
    \begin{align}
        \phi^*(v_{i_1}^{(k)},v_{-{i_1}}) =
        \begin{cases}
            X_{i_1} & \text{if } k \le k_{2}, \\
            X_{i_2} & \text{otherwise}.
        \end{cases} \label{eq:ex2:phi}
    \end{align}
    The graph $G_{i_1}(v)$ is thus constructed so that each edge from $\star$ to $v_{i_1}^{(k)}$ weighs
    \begin{align}
        v_{i_1}^{(k)}(\phi^*(v_{i_1}^{(k)},v_{-{i_1}})) =
        \begin{cases}
            v_{i_1}^{(k)}(X_{i_1}) = p_k & \text{if } k \le k_2, \\
            v_{i_1}^{(k)}(X_{i_2}) = 0   & \text{otherwise},
        \end{cases}
    \end{align}
    and each edge from $v_{i_1}^{(k')}$ to $v_{i_1}^{(k)}$ weighs
    \begin{align}
        v_{i_1}^{(k)}(\phi^*(v_{i_1}^{(k)},v_{-{i_1}})) -  v_{i_1}^{(k)}(\phi^*(v_{i_1}^{(k')},v_{-{i_1}})) =
        \begin{cases}
            p_k  & \text{if } k \le k_2 < k', \\
            -p_k & \text{if } k' \le k_2 < k, \\
            0    & \text{otherwise}.
        \end{cases}
    \end{align}
    Hence (see Figure~\ref{fig:second-price-ge}),
    $-\tau^*_{i_1}(v)$ is the distance along the path $\star\rightarrow v_{i_1}^{(k_2)}\rightarrow v_{i_1}^{(k_1)}$,
    which is equal to $p_{k_2}$,
    the second highest among bid prices.
    Otherwise, if $i_1 > i_2$, the agent $i_2$ would win against $i_1$ also if both bid the same price;
    the option rule~\eqref{eq:ex2:phi} slightly change to
    \begin{align}
        \phi^*(v_{i_1}^{(k)},v_{-{i_1}}) =
        \begin{cases}
            X_{i_1} & \text{if } k < k_{2}, \\
            X_{i_2} & \text{otherwise}.
        \end{cases}
    \end{align}
    Hence (see Figure~\ref{fig:second-price-gg}),
    $-\tau^*_{i_1}(v)$ is the distance along the path $\star\rightarrow v_{i_1}^{(k_2 - 1)}\rightarrow v_{i_1}^{(k_1)}$,
    which is equal to $p_{k_2 - 1}$,
    the lowest among strictly higher biddable prices than $p_{k_2}$.\footnote{Note $p_{k_2-1}$ becomes closer to $p_{k_2}$ as the set of biddable prices gets denser.}

    It follows after similar observations that $\tau^*_i(v) = 0,\forall i\in\mathcal{N}\setminus\{i_1\}$,
    so only $-\tau_{i_1}(v)\in\{p_{k_2-1}, p_{k_2}\}$ contributes to the revenue $-B(v;\mathcal{M}^*)$ at the auction.
    This example implies how the celebrated second-price auction~\cite{Vickrey1961} can be discretized as a special case of the proposed mechanism,
    which realizes the maximum revenue among all justified mechanisms.
\end{example}

\begin{figure}[t]
    \centering
    \begin{subfigure}{0.49\linewidth}
        \centering
        \includegraphics[width=0.7\linewidth]{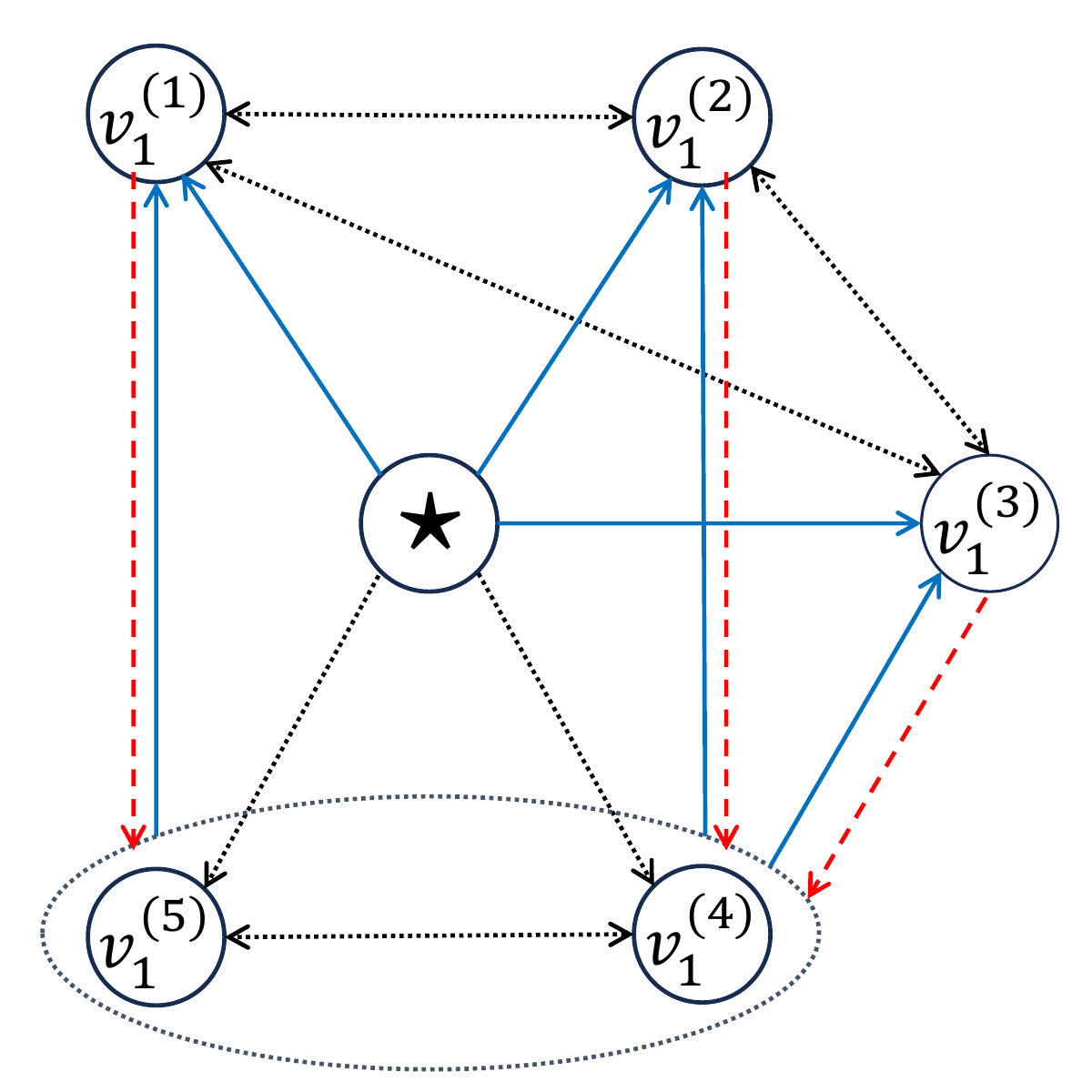}\\
        \caption{$i_1 < i_2$}
        \label{fig:second-price-ge}
    \end{subfigure}
    \begin{subfigure}{0.49\linewidth}
        \centering
        \includegraphics[width=0.7\linewidth]{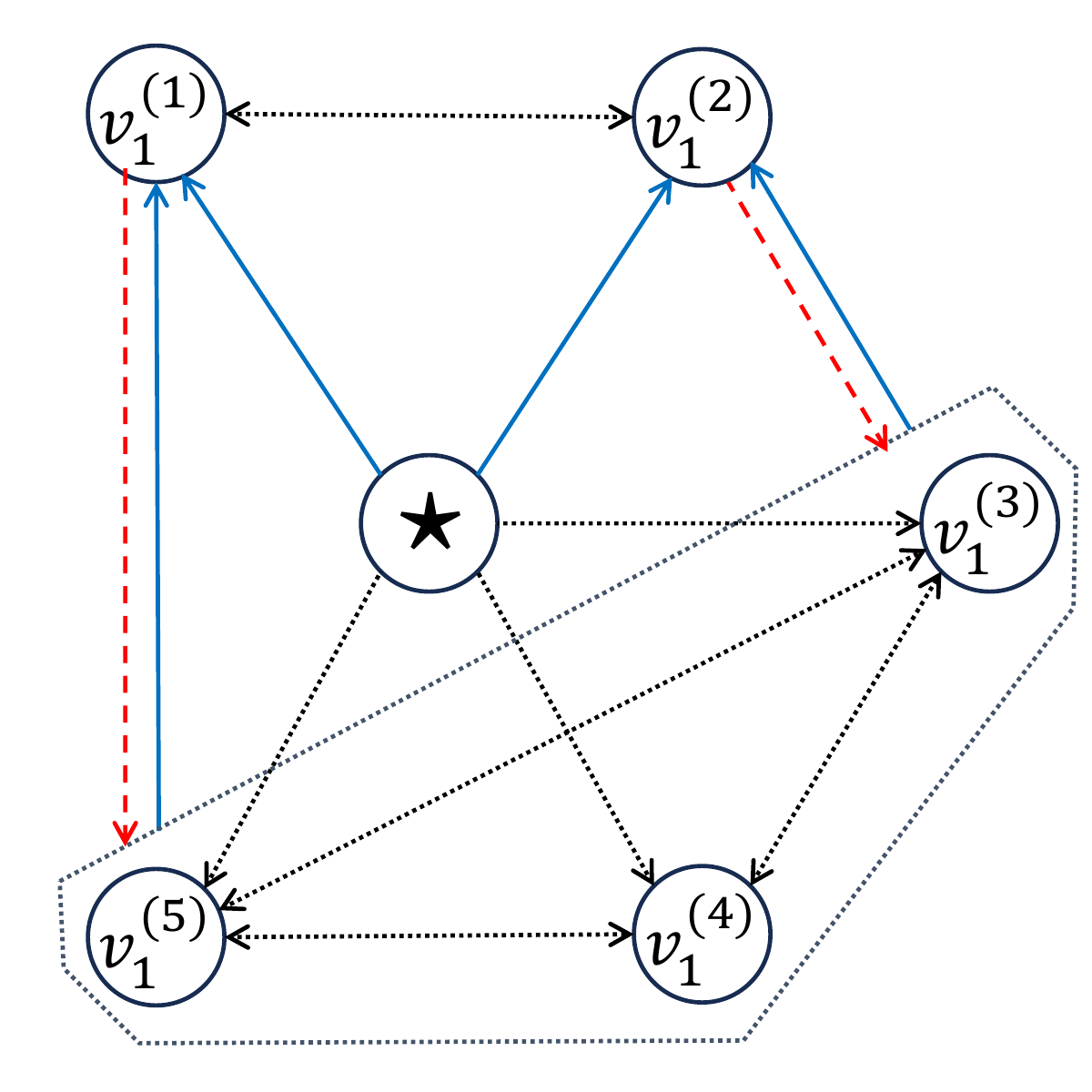}\\
        \caption{$i_1 > i_2$}
        \label{fig:second-price-gg}
    \end{subfigure}
    \caption{An example of the weighted graph $G_{i_1}(v)$ in \thref{ex:vickrey} to compute the payment from the winner $i_1 = 1$,
    where there are $d=5$ types (biddable prices) and the agent $i_2$ makes the second-highest bid of $p_3$ i.e., $k_2 = 3$.
    While not all edges are shown,
    each solid arrow represents an edge with a positive weight,
    each dashed arrow represents an edge with a negative weight,
    and each dotted arrow represents an edge with a zero weight.
    Notice that vertices inside a closed dotted line are mutually connected by zero-weight edges,
    hence they have an equal shortest distance from/to any other vertex.
    The weight of any positive edge coming into $v_1^{(k)}$ is $p_k$, and the weight of any negative edge going out of $v_1^{(k)}$ is $-p_k$.
    In the case (\ref{sub@fig:second-price-ge})~$i_1 < i_2$, one of the shortest paths from $\star$ to $v_1^{(1)}$ could be $\star\to v_1^{(3)} \to v_1^{(1)}$, along which the distance is $p_3$.
    In the other case (\ref{sub@fig:second-price-gg})~$i_1 > i_2$, one of the shortest paths from $\star$ to $v_1^{(1)}$ could be $\star\to v_1^{(2)} \to v_1^{(1)}$, along which the distance is $p_2$.
    }
    \label{fig:second-price}
\end{figure}

%% file: duality.tex
\section{LP duality} \label{sec:duality}

Let us fix an option rule $\phi$ according to SE~\eqref{eq:SE}.
Note that DSIC and IR allow us to adjust the payment $\tau_i(v)$ independently for different $i\in\mathcal{N}$ and then for different $v_{-i}\in\mathcal{V}_{-i}$.
We thus focus on deciding $(\pi(v_i)\coloneqq -\tau_i(v_i, v_{-i}))_{v_i\in\mathcal{V}_i}$ for arbitrarily fixed $i\in\mathcal{N}$ and $v_{-i}\in\mathcal{V}_{-i}$.
Using the edge weights $c$ of the graph $G_i(v)$ in Section~\ref{sec:proposed},
we consider the following (primal) LP for an arbitrary $v_i\in\mathcal{V}_i$:
\begin{align}
  \max_{\pi}\qquad \pi(v_i),              &             & \label{eq:LP_max}                                                                       \\
  \text{subject to}\qquad \pi(t) - \pi(s) & \le c(s,t), & \forall s\in\mathcal{V}_i\cup\{\star\},\forall t\in\mathcal{V}_i, \label{eq:LP_DSIC_IR} \\
  \pi(\star)                              & = 0,        & \label{eq:LP_star}
\end{align}
where \eqref{eq:LP_DSIC_IR} with \eqref{eq:LP_star} exactly represents a combination of DSIC~\eqref{eq:DSIC} and IR~\eqref{eq:IR}.
On the other hand,
the dual LP is
\begin{align}
  \min_{f} \qquad \sum_{s\in\mathcal{V}_i\cup\{\star\}, t\in\mathcal{V}_i} c(s,t)f(s,t),                   &                                                                   \\
  \text{subject to}\qquad \sum_{s\in\mathcal{V}_i\cup\{\star\}} f(s,u) - \sum_{t\in\mathcal{V}_i} f(u,t) & =
  \begin{cases}
    1 & \text{if } u = v_i, \\
    0 & \text{otherwise},
  \end{cases}
                                                                                                    & \forall u\in\mathcal{V}_i,                                        \\
  f(s,t)                                                                                            & \ge 0,
                                                                                                    & \forall s\in\mathcal{V}_i\cup\{\star\},\forall t\in\mathcal{V}_i.
\end{align}
It can be seen as the LP relaxation of the shortest path problem in the graph $G_i(v)$,
where $f(s,t)$ represents a fraction of the edge $(s,t)$ that the path passes through.
The primal LP is feasible if and only if the graph $G_i(v)$ has no negative cycle,
in which case the shortest distances $\pi^*(v_i')$ from $\star$ to each $v_i'\in\mathcal{V}_i$ in $G_i(v)$ compose an optimal solution $\pi=\pi^*$ for the primal,
while the shortest path represents an optimal solution for the dual (cf. Section 9.2 in~\cite{korte2018combopt}).
Note that it is essential in our method that the optimal solution to the primal does not depend on the choice of $v_i\in\mathcal{V}_i$,
for which $\pi(v_i) = -\tau_i(v_i, v_{-i})$ is maximized \eqref{eq:LP_max}.

%% file: ama.tex
\section{Relation to the affine maximizer auction (AMA)} \label{sec:ama}


With regard to the affine maximizer auction (AMA) mentioned in Section~\ref{sec:related},
this section provides how the proposed mechanism can be applied in that context,
as well as a brief introduction of basic concepts about AMA.
The environment is now assumed to allow no negative value~\eqref{eq:positive_type},
as is the case in single-side auctions,
and possibly infinite type domains unlike our model.
Then we choose an affine maximizer in the following \thref{def:affine_maximizer} as the option rule,
which may not necessarily satisfy SE~\eqref{eq:SE},
the stronger requirement than \eqref{eq:affine_maximizer}.

\begin{definition}[Affine maximizer] \thlabel{def:affine_maximizer}
    For an environment $\mathcal{E}=(\mathcal{N}, \mathcal{X}, \mathcal{V})$,
    an option rule $\phi$ of a mechanism is said to be an \textbf{affine maximizer} if and only if
    \begin{align}
        \phi(v) & \in \arg\max_{X\in\mathcal{X}} \left\{\sum_{i\in\mathcal{N}} w_i v_i(X) + \lambda(X)\right\},
                & \forall v\in\mathcal{V} \label{eq:affine_maximizer}
    \end{align}
    holds for some agent weights $w = (w_i)_{i\in\mathcal{N}}\in\mathbb{R}_{>0}^{\mathcal{N}}$ and some option weight $\lambda:\mathcal{X}\to\mathbb{R}$.
\end{definition}

\thref{thm:Roberts} states a notable fact that motivates mechanisms to adopt affine maximizers~\eqref{eq:affine_maximizer}.
Moreover,
under certain assumptions,
even combinatorial auctions or other restricted environments
allow only ``almost affine maximizers'' to guarantee DSIC~\cite{Lavi2003}.

\begin{theorem}[Roberts~\cite{Roberts1979}] \thlabel{thm:Roberts}
    If all type domains are unrestricted i.e.,
    $\mathcal{V}_i = \mathbb{R}^{\mathcal{X}}$,
    and if there are at least three options i.e.,
    $|\mathcal{X}| \ge 3$,
    then the option rule of any mechanism $(\phi,\tau)$ that satisfies DSIC~\eqref{eq:DSIC}
    must be an affine maximizer for some weights $(w,\lambda)$.
\end{theorem}

Such an affine maximizer $\phi$ as the option rule
often comes with a payment rule $\tau$ in the following generic form~\eqref{eq:weighted_VCG},
which actually enjoys DSIC~\eqref{eq:DSIC}.
We thus obtain a weighted extension of the VCG mechanism (\thref{def:VCG}).

\begin{definition}[Weighted VCG mechanism] \thlabel{def:weighted_VCG}
    A mechanism $\mathcal{M} = (\phi, \tau)$ is called the \textbf{weighted Vickrey-Clarke-Groves (VCG) mechanism} if and only if the option rule $\phi$ is an affine maximizer for some weights $(w, \lambda)$ and $\mathcal{M}$ satisfies
    \begin{align}
        \tau_i(v) & = \frac{1}{w_i}\left(\sum_{j\in\mathcal{N}\setminus\{i\}}w_j v_j(\phi(v)) + \lambda(\phi(v)) \right) - h_i(v_{-i}),
                  & \forall i\in\mathcal{N} \label{eq:weighted_VCG}
    \end{align}
    with some functions $(h_i:\mathcal{V}_{-i}\rightarrow\mathbb{R})i\in\mathcal{N}$.
\end{definition}

\begin{theorem}[Proposition 1.31 in \cite{Nisan2007}]
    Every weighted VCG mechanism satisfies DSIC~\eqref{eq:DSIC}.
\end{theorem}

Furthermore,
the following \thref{thm:VCG_unique} implies that this VCG family is the unique truthful solution in environments with a wide class of continuous domains.

\begin{theorem}[Theorem 1.37 in \cite{Nisan2007}] \thlabel{thm:VCG_unique}
    If every participant $i$'s type domain $\mathcal{V}_i$ is connected in the $L^2$ space,
    and if
    both two mechanism $(\phi, \tau)$ and $(\phi, \tau')$ satisfy DSIC~\eqref{eq:DSIC},
    then
    \begin{align}
        \tau'_i(v) - \tau_i(v) & = h_i(v_{-i}),
                               & \forall i\in\mathcal{N},\forall v\in\mathcal{V},
    \end{align}
    with some functions $(h_i:\mathcal{V}_{-i}\rightarrow\mathbb{R})i\in\mathcal{N}$.
\end{theorem}

The Clarke pivot rule (Definition~\ref{def:Clarke}) is similarly extended as follows,
which provides a mechanism called the affine maximizer auction (AMA).
Most studies about maximizing revenue treat the weights $(w, \lambda)$ as parameters of AMA and optimize them to increase the revenue in analytical or heuristic approach.

\begin{definition}[AMA] \thlabel{def:Clarke}
    The \textbf{affine maximizer auction (AMA)} specifies the payment rule~\eqref{eq:weighted_VCG} of a VCG mechanism as
    \begin{align}
        h_i(v_{-i}) & = \frac{1}{w_i}\max_{X\in\mathcal{X}} \left(\sum_{j\in\mathcal{N}\setminus\{i\}}w_j v_j(X) + \lambda(X) \right),
                    & \forall i\in\mathcal{N},\forall v\in\mathcal{V}.
    \end{align}
\end{definition}

\begin{lemma} \thlabel{lem:AMA_IR}
    The AMA $\mathcal{M} = (\phi,\tau)$ makes every agent pay to the broker:
    \begin{align}
        \tau_i(v) & \le 0,
                  & \forall i\in\mathcal{N},\forall v\in\mathcal{V}.\label{eq:ama_no_comp}
    \end{align}
    It satisfies IR~\eqref{eq:IR} if all possible types are non-negative~\eqref{eq:positive_type}.
\end{lemma}

\begin{proof}
    Recall that $w_i > 0,\forall i\in\mathcal{N}$.
    The first claim~\eqref{eq:ama_no_comp} follows by definition:
    \begin{align}
        & w_i\tau_i(v) \nonumber \\
        & = \sum_{j\in\mathcal{N}\setminus\{i\}} w_j v_j(\phi(v)) + \lambda(\phi(v)) - \max_{X\in\mathcal{X}} \left(\sum_{j\in\mathcal{N}\setminus\{i\}} w_j v_j(X) + \lambda(X)\right) \nonumber                                                                                   \\
                     & \le 0,                                                                                                                                                                                  & \forall i\in\mathcal{N},\forall v\in\mathcal{V}.\label{eq:externality_weighted}
    \end{align}
    It is also observed that
    \begin{align}
        & w_i u_i(v;\mathcal{M}) \nonumber \\
        & = w_iv_i(\phi(v)) + w_i\tau_i(v) \nonumber                                                                                              \\
                           & = \sum_{j\in\mathcal{N}} w_j v_j(\phi(v)) + \lambda(\phi(v)) - h_i(v_{-i}) \nonumber                                              \\
                           & = \max_{X\in\mathcal{X}} \left(\sum_{j\in\mathcal{N}} w_j v_j(X) + \lambda(X) \right) 
                           - \max_{X\in\mathcal{X}} \left( \sum_{j\in\mathcal{N}\setminus\{i\}} w_j v_j(X) + \lambda(X) \right) \nonumber \\
                           & \ge \inf_{X\in\mathcal{X}} w_i v_i(X),
                           & \forall i\in\mathcal{N},v\in\mathcal{V},
    \end{align}
    which yields IR~\eqref{eq:IR} when \eqref{eq:positive_type} holds.
\end{proof}




We show that the proposed mechanism (Algorithm~\ref{alg:main}) satisfies DSIC~\eqref{eq:DSIC} and IR~\eqref{eq:IR},
even when it is built with any affine maximizer~\eqref{eq:affine_maximizer} as follows.
Based on \thref{cor:no-neg-cyc-weighted}, an extension of \thref{lem:no-neg-cyc},
\thref{cor:justified_weighted} and \thref{cor:maximal_weighted} are derived from exactly the same discussions to obtain \thref{thm:justified} and \thref{thm:maximal}, respectively.
In particular,
\thref{cor:maximal_weighted} suggests that the proposed mechanism may boost the revenue after the optimal weights are found by AMA.

\begin{corollary} \thlabel{cor:no-neg-cyc-weighted}
    For any $v\in\mathcal{V}$ and $i\in\mathcal{N}$,
    the weighted directed graph $G_i(v)$ has no negative directed cycle,
    if only the option rule $\phi^*$ is an affine maximizer.
\end{corollary}

\begin{proof}
    Suppose the proposed mechanism is built with any affine maximizer $\phi^*$ for some weights $(w,\lambda)$.
    Fix an arbitrary pair of $i\in\mathcal{N}$ and $v\in\mathcal{V}$
    for which we are going to show the claim holds.
    Note that no cycle contains the auxiliary vertex $\star$,
    since no edge comes out from it.
    Consider an arbitrary sequence $v_i^{(1)},\ldots,v_i^{(k)}$ in $\mathcal{V}_i$,
    and let $v_i^{(0)} \coloneqq v_i^{(k)}$.
    Let also $v^{(\ell)}\coloneqq (v_i^{(\ell)}, v_{-i}),\forall\ell\in\{0,1,\ldots,k\}$.
    By definition of \textproc{SelectOption},
    it follows that
    \begin{align}
         & w_i v_i^{(\ell)}(\phi^*(v^{(\ell)})) + \sum_{j\in\mathcal{N}\setminus\{i\}} w_j v_{j}(\phi^*(v^{(\ell)})) + \lambda(\phi^*(v^{(\ell)})) \nonumber  \\
         & \ge w_i v_i^{(\ell)}(\phi^*(v^{(\ell-1)})) + \sum_{j\in\mathcal{N}\setminus\{i\}} w_j v_{j}(\phi^*(v^{(\ell-1)})) + \lambda(\phi^*(v^{(\ell-1)})),
         & \forall\ell\in\{1,\ldots,k\}. \label{eq:maxop-weighted}
    \end{align}
    Summing up \eqref{eq:maxop-weighted} over all $\ell\in\{1,\ldots,k\}$ yields
    \begin{align}
         & w_i\sum_{\ell=1}^k (v_i^{(\ell)}(\phi^*(v^{(\ell)})) - v_i^{(\ell)}(\phi^*(v^{(\ell-1)})))      \nonumber                                                         \\
         & \ge \sum_{\ell=1}^k \left( \sum_{j\in\mathcal{N}\setminus\{i\}} w_j v_{j}(\phi^*(v^{(\ell-1)})) + \lambda(\phi^*(v^{(\ell-1)}))\right.  \nonumber \\
         & \quad  \left. -\sum_{j\in\mathcal{N}\setminus\{i\}} w_j v_{j}(\phi^*(v^{(\ell)})) - \lambda(\phi^*(v^{(\ell)})) \right) \nonumber                              \\
         & = 0, \label{eq:nonnegative-cycle-weighted}
    \end{align}
    where the last equality follows from $v^{(0)} = v^{(k)}$ by definition.
    Since $w_i > 0$,
    Eq.~\eqref{eq:nonnegative-cycle-weighted} implies that the sum of the edge weights along any cycle in $G_i(v)$ is not negative.
\end{proof}

\begin{theorem} \thlabel{cor:justified_weighted}
    The proposed mechanism built with an affine maximizer satisfies DSIC~\eqref{eq:DSIC} and IR~\eqref{eq:IR}.
\end{theorem}

\begin{corollary} \thlabel{cor:maximal_weighted}
    The proposed mechanism $\mathcal{M}^*$ built with an affine maximizer $\phi^*$ satisfies,
    against any mechanism $\mathcal{M} = (\phi^*, \tau)$ that enjoys DSIC~\eqref{eq:DSIC} and IR~\eqref{eq:IR},
    \begin{align}
        \tau^*_i(v)        & \le \tau_i(v),        & \forall i\in\mathcal{N},\forall v\in\mathcal{V}, \label{eq:tau_maximized_weighted} \\
        B(v;\mathcal{M}^*) & \le B(v;\mathcal{M}), & \forall v\in\mathcal{V}. \label{eq:rev_maximized_weighted}
    \end{align}
\end{corollary}

%% file: more_exp.tex
\section{Additional experiments}
\label{sec:more_exp}

In this section, we provide additional results of experiments to confirm that the observations made with Figure~\ref{fig:exp} hold with other choices of parameters.  We use the same experimental settings as discussed in Section~\ref{sec:exp}.

\begin{figure}[tb]
    \centering
    $|\mathcal{N}|=8$\\
    \begin{subfigure}{0.245\linewidth}
        \includegraphics[width=\linewidth]{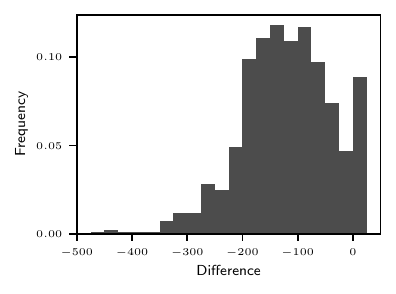}
    \end{subfigure}
    \begin{subfigure}{0.245\linewidth}
        \includegraphics[width=\linewidth]{figs/plot_diff_vs_n_32_256_16_-100_100_1000_stddev.pdf}
    \end{subfigure}
    \begin{subfigure}{0.245\linewidth}
        \includegraphics[width=\linewidth]{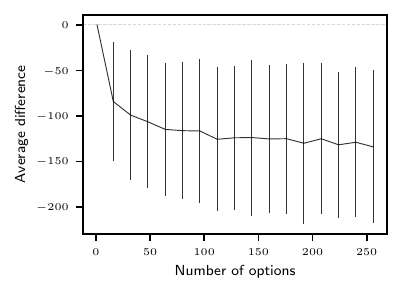}
    \end{subfigure}
    \begin{subfigure}{0.245\linewidth}
        \includegraphics[width=\linewidth]{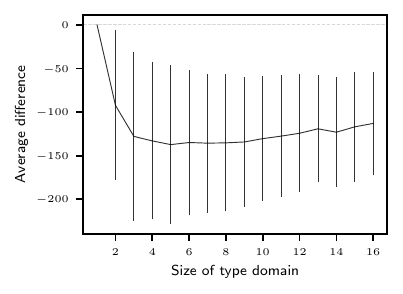}
    \end{subfigure}
    \ \\
    \ \\
    $|\mathcal{N}|=16$\\
    \begin{subfigure}{0.245\linewidth}
        \includegraphics[width=\linewidth]{figs/hist_diff_n_16_256_16_-100_100_1000_False.pdf}
    \end{subfigure}
    \begin{subfigure}{0.245\linewidth}
        \includegraphics[width=\linewidth]{figs/plot_diff_vs_n_32_256_16_-100_100_1000_stddev.pdf}
    \end{subfigure}
    \begin{subfigure}{0.245\linewidth}
        \includegraphics[width=\linewidth]{figs/plot_diff_vs_m_16_256_16_-100_100_1000_stddev.pdf}
    \end{subfigure}
    \begin{subfigure}{0.245\linewidth}
        \includegraphics[width=\linewidth]{figs/plot_diff_vs_d_16_256_16_-100_100_1000_stddev.pdf}
    \end{subfigure}
    \ \\
    \ \\
    $|\mathcal{N}|=32$\\
    \begin{subfigure}{0.245\linewidth}
        \includegraphics[width=\linewidth]{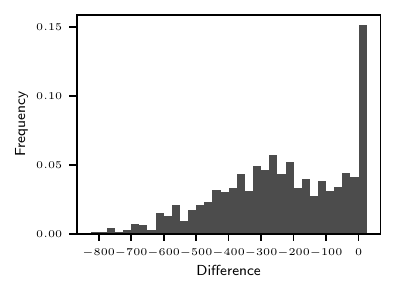}
        \caption{$|\mathcal{N}|=8, 16, 32$}
        \label{fig:expn:hist}
    \end{subfigure}
    \begin{subfigure}{0.245\linewidth}
      \includegraphics[width=\linewidth]{figs/plot_diff_vs_n_32_256_16_-100_100_1000_stddev.pdf}
      \caption{$1\le|\mathcal{N}|\le32$}
      \label{fig:expn:plot_vs_n}
    \end{subfigure}
    \begin{subfigure}{0.245\linewidth}
      \includegraphics[width=\linewidth]{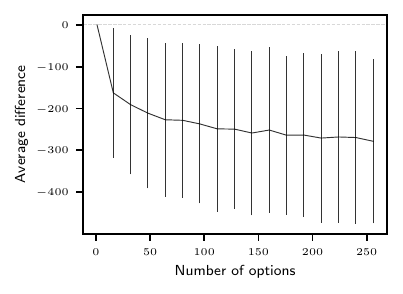}
      \caption{$1\le|\mathcal{X}|\le256$}
      \label{fig:expn:plot_vs_m}
    \end{subfigure}
    \begin{subfigure}{0.245\linewidth}
      \includegraphics[width=\linewidth]{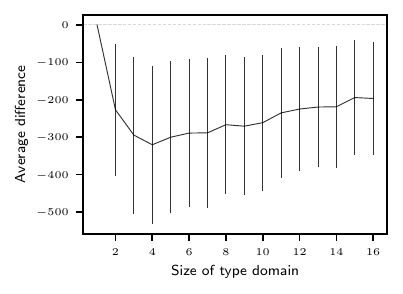}
      \caption{$1\le|\mathcal{V}_i|\le16$}
      \label{fig:expn:plot_vs_d}
    \end{subfigure}
    \caption{
      Difference in budget required by the proposed mechanism relative to VCG-budget, shown as
      a histogram in (\ref{sub@fig:expn:hist})
      as well as the average against the number of agents $|\mathcal{N}|$ in (\ref{sub@fig:expn:plot_vs_n}), 
      the number of options $|\mathcal{X}|$ in (\ref{sub@fig:expn:plot_vs_m}), and      
      the size of type domain $|\mathcal{V}_i|$ in (\ref{sub@fig:expn:plot_vs_d}),
      where $|\mathcal{N}|$ is fixed at 8, 16, or 32, as indicated in each row, except in (\ref{sub@fig:expn:plot_vs_n}).
    }
    \label{fig:expn}
\end{figure}

In Figure~\ref{fig:expn}, we vary the value of $|\mathcal{N}|$, which is fixed at $|\mathcal{N}|=16$ in Figure~\ref{fig:exp}.  Specifically, we let $|\mathcal{N}|=8$ in the top row, $|\mathcal{N}|=16$ in the middle row, and $|\mathcal{N}|=32$ in the bottom row.  Hence, the middle row is identical to Figure~\ref{fig:exp}.  Also, all of the three panels in Figure~\ref{fig:expn:plot_vs_n} are identical to Figure~\ref{fig:exp:plot_vs_n}, since $|\mathcal{N}|$ is varied identically in these panels.  Notice the difference in the scale of each panel.

Overall, we reconfirm that the proposed mechanism requires strictly lower budget than VCG-budget for a large fraction of instances in the settings of Figure~\ref{fig:expn}.  Specifically, Figure~\ref{fig:expn:hist} shows that the proposed mechanism strictly improves upon VCG-budget in 91.1\% of the instances with $|\mathcal{N}|=8$, and 88.3\% with $|\mathcal{N}|=16$, and 84.9\% with $|\mathcal{N}|=32$.  In Figure~\ref{fig:expn:plot_vs_n}-\ref{fig:expn:plot_vs_d}, we can also reconfirm that the relative advantage of the proposed mechanism with respect to the average difference in the budget tends to increase with $|\mathcal{N}|$ (observe the changes in the scale from the top row to the bottom).

\begin{figure}[tb]
    \centering
    $v_i(X)\in[-1,1]$\\
    \begin{subfigure}{0.245\linewidth}
        \includegraphics[width=\linewidth]{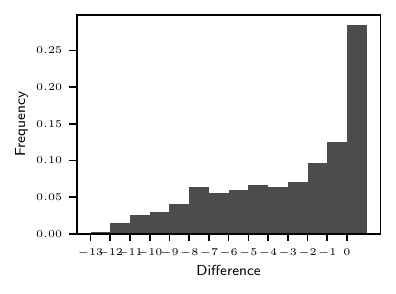}
    \end{subfigure}
    \begin{subfigure}{0.245\linewidth}
        \includegraphics[width=\linewidth]{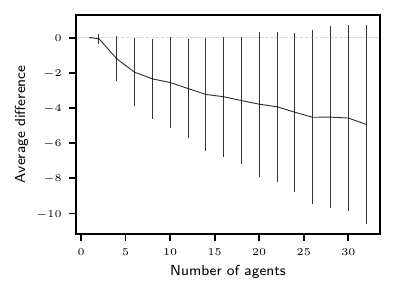}
    \end{subfigure}
    \begin{subfigure}{0.245\linewidth}
        \includegraphics[width=\linewidth]{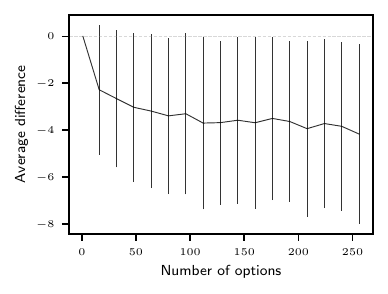}
    \end{subfigure}
    \begin{subfigure}{0.245\linewidth}
        \includegraphics[width=\linewidth]{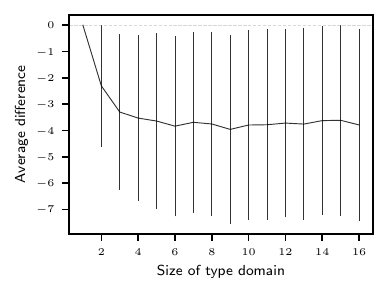}
    \end{subfigure}
    \ \\
    \ \\
    $v_i(X)\in[-10,10]$\\
    \begin{subfigure}{0.245\linewidth}
        \includegraphics[width=\linewidth]{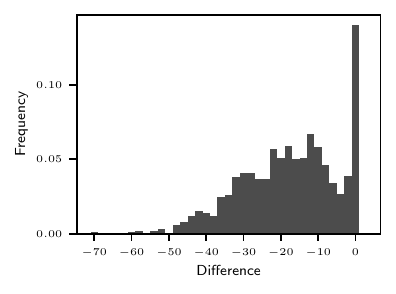}
    \end{subfigure}
    \begin{subfigure}{0.245\linewidth}
        \includegraphics[width=\linewidth]{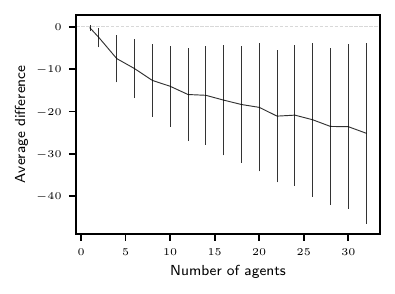}
    \end{subfigure}
    \begin{subfigure}{0.245\linewidth}
        \includegraphics[width=\linewidth]{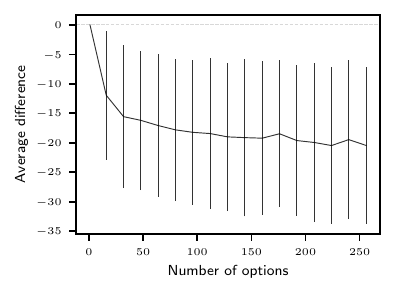}
    \end{subfigure}
    \begin{subfigure}{0.245\linewidth}
        \includegraphics[width=\linewidth]{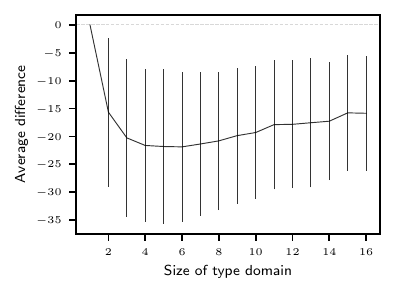}
    \end{subfigure}
    \ \\
    \ \\
    $v_i(X)\in[-100,100]$\\
    \begin{subfigure}{0.245\linewidth}
        \includegraphics[width=\linewidth]{figs/hist_diff_n_16_256_16_-100_100_1000_False.pdf}
    \end{subfigure}
    \begin{subfigure}{0.245\linewidth}
        \includegraphics[width=\linewidth]{figs/plot_diff_vs_n_32_256_16_-100_100_1000_stddev.pdf}
    \end{subfigure}
    \begin{subfigure}{0.245\linewidth}
        \includegraphics[width=\linewidth]{figs/plot_diff_vs_m_16_256_16_-100_100_1000_stddev.pdf}
    \end{subfigure}
    \begin{subfigure}{0.245\linewidth}
        \includegraphics[width=\linewidth]{figs/plot_diff_vs_d_16_256_16_-100_100_1000_stddev.pdf}
    \end{subfigure}
    \ \\
    \ \\
    $v_i(X)\in[-1000,1000]$\\
    \begin{subfigure}{0.245\linewidth}
        \includegraphics[width=\linewidth]{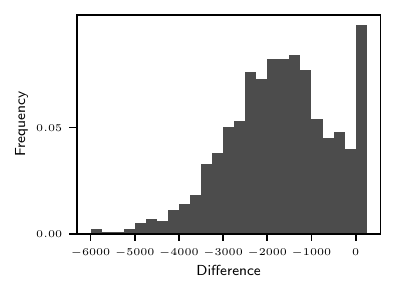}
        \caption{$|\mathcal{N}|=16$}
        \label{fig:expv:hist}
    \end{subfigure}
    \begin{subfigure}{0.245\linewidth}
      \includegraphics[width=\linewidth]{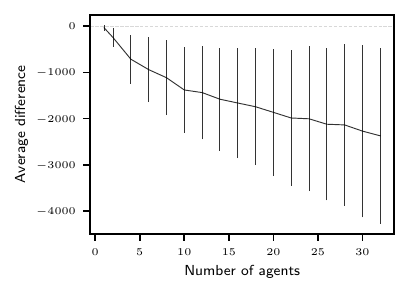}
      \caption{$1\le|\mathcal{N}|\le32$}
      \label{fig:expv:plot_vs_n}
    \end{subfigure}
    \begin{subfigure}{0.245\linewidth}
      \includegraphics[width=\linewidth]{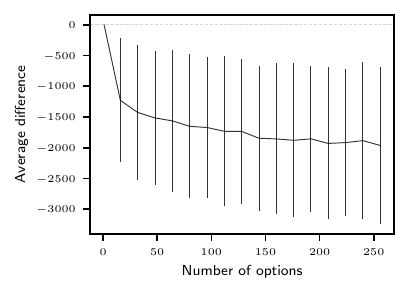}
      \caption{$1\le|\mathcal{X}|\le256$}
      \label{fig:expv:plot_vs_m}
    \end{subfigure}
    \begin{subfigure}{0.245\linewidth}
      \includegraphics[width=\linewidth]{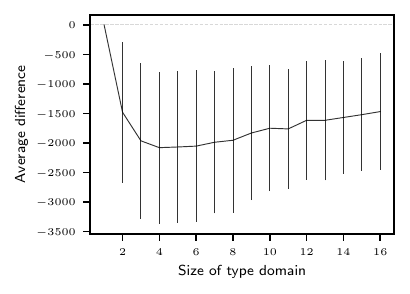}
      \caption{$1\le|\mathcal{V}_i|\le16$}
      \label{fig:expv:plot_vs_d}
    \end{subfigure}
    \caption{
      Difference in budget required by the proposed mechanism relative to VCG-budget, shown as
      a histogram in (\ref{sub@fig:expv:hist})
      as well as the average against the number of agents $|\mathcal{N}|$ in (\ref{sub@fig:expv:plot_vs_n}), 
      the number of options $|\mathcal{X}|$ in (\ref{sub@fig:expv:plot_vs_m}), and      
      the size of type domain $|\mathcal{V}_i|$ in (\ref{sub@fig:expv:plot_vs_d}),
      where $|\mathcal{N}|=16$ is fixed except in (\ref{sub@fig:expv:plot_vs_n}).
      Here, we vary the range of $v_i(X)$ as indicated in each row.
    }
    \label{fig:expv}
\end{figure}

In Figure~\ref{fig:expv}, we now vary the support of the uniform distribution on integers\footnote{The integer support is chosen to ensure that all numerical data and computations are handled at unlimited precision.} from which $v_i(X)$ is sampled.  Note that, although changes in the \emph{scale} of the support of $v_i(X)$ would equally change the scale of the budgets required by the proposed mechanism and VCG-budget, changes in the \emph{range of the integers} would change the number of unique values that $v_i(X)$'s can take, which in turn can affect the relative performance of the mechanisms (recall the arguments in the proof of \thref{thm:time}).  Specifically, we let the range of the integers be $[-1, 1]$ in the first row, $[-10, 10]$ in the second row, $[-100, 100]$ in the third row, and $[-1000, 1000]$ in the bottom row.  Hence, the third row is identical to Figure~\ref{fig:exp}.  Again, notice the difference in the scale of each panel.

Overall, we reconfirm that the proposed mechanism requires strictly lower budget than VCG-budget for a large fraction of instances for all of the ranges studied in Figure~\ref{fig:expv}.  Taking a closer look, we can observe in Figure~\ref{fig:expv:hist} that there are relatively infrequent opportunities for the proposed mechanism to improve upon VCG-budget when there are only a few unique values of $v_i(X)$ (top row).  Specifically, the proposed mechanism strictly improves upon VCG-budget in 71.6\% of the instances when $v_i(X)$ can only take an integer value in $[-1, 1]$, while this fraction increases to 86.0\%, 88.3\%, and 90.2\% when the range changes to $[-10, 10]$, $[-100, 100]$, and $[-1000, 1000]$, respectively.  We can also make analogous observations in Figure~\ref{fig:expv:plot_vs_n}-\ref{fig:expv:plot_vs_d}.  While some of the error bars contains 0 in the top row, this is simply due to the skewed distribution of the differences as shown in Figure~\ref{fig:expv:hist} and does not imply that the average difference may be positive.  \thref{thm:maximal} guarantees that the proposed mechanism requires no larger budget than VCG-budget \emph{for every instance}.


%% file: redistribution.tex
\section{Redistribution of revenue} \label{sec:redist}


First,
let us formally define the budget balance conditions mentioned in Section~\ref{sec:related},~\ref{sec:model} as follows.

\begin{definition} \thlabel{def:BB}
  For an environment $\mathcal{E}=(\mathcal{N}, \mathcal{X}, \mathcal{V})$,
  a mechanism $\mathcal{M} = (\phi, \tau)$ is said to satisfy
  \textbf{Weak Budget Balance (WBB)} if and only if
  \begin{align}
     & B(v; \mathcal{M}) \le 0, & \forall v\in\mathcal{V}; \label{eq:WBB}
  \end{align}
  \textbf{Strong Budget Balance (SBB)} if and only if
  \begin{align}
     & B(v; \mathcal{M}) = 0, & \forall v\in\mathcal{V}. \label{eq:SBB}
  \end{align}
\end{definition}

With regard to redistribution of revenue mentioned in Section \ref{sec:related},
the following \thref{cor:redist} constructs one possible algorithm that can adapt the proposed mechanism
to return part of the surplus (negative budget) back to agents
while any positive budget are still minimized.

\begin{lemma} \thlabel{lem:redist}
  For any justified mechanism $\mathcal{M} = (\phi, \tau)$,
  a mechanism $\mathcal{M}' = (\phi, \tau')$ is also justified if
  \begin{align}
    \tau'_i(v) = \tau_i(v) + h_i(v_{-i}),\quad\forall i\in\mathcal{N},\forall v\in\mathcal{V}
  \end{align}
  holds with some function $h_i:\mathcal{V}_{-i}\to\mathbb{R}_{\ge 0}$ for each $i\in\mathcal{N}$.
\end{lemma}

\begin{proof}
  The claim straightforwardly follows from \thref{def:justified}.
\end{proof}

\begin{corollary} \thlabel{cor:redist}
  For any justified mechanism  $\mathcal{M} = (\phi, \tau)$,
  there exists a justified mechanism $\mathcal{M}' = (\phi, \tau')$ that satisfies the following:
  \begin{align}
    \tau'_i(v)                                            & \ge \tau_i(v),                   & \forall i\in \mathcal{N},\forall v\in\mathcal{V},          \\
    B(v;\mathcal{M}')                                     & \le \max\{0, B(v;\mathcal{M})\}, & \forall v\in\mathcal{V},                                   \\
    \max_{v_i\in\mathcal{V}_i} B(v_i,v_{-i};\mathcal{M}') & \ge 0,                           & \forall v_{-i}\in\mathcal{V}_{-i},\forall i\in\mathcal{N}.
  \end{align}
\end{corollary}

\begin{proof}
  For convenience,
  let us have agents $\mathcal{N} \coloneqq \{1,\ldots,|\mathcal{N}|\}$ indexed.
  Given any justified mechanism $\mathcal{M}\eqqcolon\mathcal{M}^{(0)} = (\phi, \tau^{(0)})$,
  we recursively define the mechanism $\mathcal{M}^{(t)}=(\phi, \tau^{(t)})$ as follows:
  \begin{align}
    \tau^{(t)}_i(v) = \begin{cases}
                        \tau^{(t-1)}_i(v) + \max\left\{0, \displaystyle\min_{v_i'\in\mathcal{V}_i}-B(v_i',v_{-i};\mathcal{M}^{(t-1)}) \right\} & \text{if } i = t, \\
                        \tau^{(t-1)}_i(v)                                                                                                      & \text{otherwise,}
                      \end{cases} \nonumber \\
    \forall i\in\mathcal{N},\forall v\in\mathcal{V},\forall t\in\mathcal{N}.
  \end{align}
  Then every $\mathcal{M}^{(t)}$ is justified because of \thref{lem:redist} and satisfies the following:
  \begin{align}
    \tau^{(t)}_i(v)                                            & \ge \tau^{(t-1)}_i(v),                  & \forall i\in\mathcal{N}, \forall v\in\mathcal{V},\forall t\in\mathcal{N},             \\
    B(v;\mathcal{M}^{(t)})                                     & \le \max\{0,B(v;\mathcal{M}^{(t-1)})\}, & \forall v\in\mathcal{V},\forall t\in\mathcal{N},                                      \\
    \max_{v_i\in\mathcal{V}_i} B(v_i,v_{-i};\mathcal{M}^{(t)}) & \ge 0,                                  & \forall v_{-i}\in\mathcal{V}_{-i},\forall i\in\{1,\ldots,t\},\forall t\in\mathcal{N}.
  \end{align}
  Therefore,
  the mechanism $\mathcal{M}'\coloneqq\mathcal{M}^{(|\mathcal{N}|)}$ finally satisfies all the desired conditions.
\end{proof}

Note that the proposed mechanism (Algorithm~\ref{alg:main}) coupled with some redistribution technique,
the one in \thref{cor:redist} or other methods,
offers a way to \textit{approximately} achieve SBB~\eqref{eq:SBB} while satisfying SE~\eqref{eq:SE}, DSIC~\eqref{eq:DSIC}, and IR~\eqref{eq:IR}.
When WBB~\eqref{eq:WBB} is unattainable under those three conditions,
the proposed mechanism yields the minimum budget with no redistribution applied (i.e., SBB is approximated in the best possible manner).
Otherwise,
if WBB can be achieved,
the proposed mechanism may result in positive revenue,
which is then discounted by the redistribution to better approximate SBB.
It remains future work to establish an \textit{optimal} redistribution method for the proposed mechanism just as studied for the VCG mechanisms.

%% file: impact.tex
\section{Broader impacts}
\label{sec:impact}

Our approach can have positive societal impacts by minimizing the budget needed to realize socially efficient mechanisms for a number of environments including trading networks, double-sided auctions in cloud markets, workforce management, federated learning, and cloud-sourcing, as we have discussed in Section~\ref{sec:intro}.  This however does not mean that all the agents get equally benefit from our approach, even though our approach guarantees individual rationality.  When our approach is applied, it is thus recommended to carefully assess whether such fairness needs to be considered and take necessary actions to mitigate unfairness if needed.